\documentclass[reprint,superscriptaddress,amsmath,amssymb,aps,pra,]{revtex4-2}

\usepackage{graphicx}
\usepackage{floatrow}

\usepackage{dcolumn}% Align table columns on decimal point
\usepackage{bm}% bold math

\usepackage{amsfonts}
\usepackage{dsfont}
\usepackage{bm}
\usepackage{bbold}
\usepackage{xcolor}
\usepackage{lipsum,babel}
\usepackage[normalem]{ulem}

\usepackage[utf8]{inputenc}
\usepackage{amsmath}
\usepackage{physics}
\usepackage{hyperref}
\usepackage{amssymb}
\usepackage{braket}
\usepackage{mathtools}
\usepackage{color}
\usepackage{amsthm}

\theoremstyle{theorem}

\newtheorem{theorem}{Theorem}

\newtheorem{conjecture}{Conjecture}

\DeclareMathOperator*{\E}{\mathbb{E}}

\DeclareMathOperator*{\poly}{\text{poly}}

\definecolor{red}{RGB}{255,0,0}
\definecolor{navy}{RGB}{46,72,102}
\definecolor{pink}{RGB}{219,48,122}
\definecolor{grey}{RGB}{184,184,184}
\definecolor{yellow}{RGB}{255,192,0}
\definecolor{grey1}{RGB}{217,217,217}
\definecolor{grey2}{RGB}{166,166,166}
\definecolor{grey3}{RGB}{89,89,89}

\newcommand{\co}[1]{{\color{blue}{(CO: #1)}}}

\begin{document}
\preprint{APS/123-QED}

\title{Exploring Shallow-Depth Boson Sampling: Towards Scalable Quantum Supremacy}
\author{Byeongseon Go}
%\email{gbs1997@snu.ac.kr}
\affiliation{Department of Physics and Astronomy, Seoul National University, Seoul 08826, Republic of Korea}
\author{Changhun Oh}
\email{changhun0218@gmail.com}
\affiliation{Pritzker School of Molecular Engineering, The University of Chicago, Chicago, IL 60637, USA}
\affiliation{Department of Physics, Korea Advanced Institute of Science and Technology, Daejeon 34141, Republic of Korea}
\author{Liang Jiang}
\affiliation{Pritzker School of Molecular Engineering, The University of Chicago, Chicago, IL 60637, USA}
\author{Hyunseok Jeong}
%\email{h.jeong37@gmail.com}
\affiliation{Department of Physics and Astronomy, Seoul National University, Seoul 08826, Republic of Korea}

\begin{abstract}
Boson sampling is a sampling task proven to be hard to simulate efficiently using classical computers under plausible assumptions, which makes it an appealing candidate for quantum supremacy. 
However, due to a large noise rate for near-term quantum devices, it is still unclear whether those noisy devices maintain the quantum advantage for much larger quantum systems.
Since the noise rate typically grows with the circuit depth, an alternative is to find evidence of simulation hardness at the shallow-depth quantum circuit. 
To find the evidence, one way is to identify the minimum depth required for the average-case hardness of approximating output probabilities, which is considered a necessary condition for the state-of-the-art technique to prove the simulation hardness of boson sampling.
In this work, we analyze the output probability distribution of shallow-depth boson sampling for Fock-states and Gaussian states and examine the limitation of the average-case hardness argument at this shallow-depth regime for geometrically local architectures. 
We propose a shallow-depth linear optical circuit architecture that can overcome the problems associated with geometrically local architectures. 
Our numerical results suggest that this architecture demonstrates possibilities of average-case hardness properties in a shallow-depth regime through its resemblance to the global Haar-random boson sampling circuit. 
This result implies that the corresponding architecture has the potential to be utilized for scalable quantum supremacy with its shallow-depth boson sampling. 
\end{abstract}

\maketitle

\section{Introduction}
We live in an exciting era with the emergence of noisy intermediate-scale quantum (NISQ) devices~\cite{Preskill2018NISQ}. 
Those NISQ devices are expected to be able to outperform classical computers in some computational tasks, which refers to quantum advantage or quantum supremacy.
Various computational problems are proposed to be a candidate for demonstrating quantum advantage with NISQ devices, and the prominent candidates at the present time are boson sampling (BS) \cite{aaronson2011computational, hamilton2017gaussian, deshpande2021quantum} and random circuit sampling (RCS) \cite{Bouland2019nature, arute2019quantum}; both denote sampling problems from the random quantum circuit instances but in different systems. 
They have been complexity-theoretically proven to be hard to efficiently simulate with classical computers under plausible assumptions. 
Ever since the theoretical foundation~\cite{aaronson2011computational, Bouland2019nature, hamilton2017gaussian, deshpande2021quantum}, there have been plenty of experimental results claiming the first realization of quantum advantage with BS~\cite{zhong2020quantum, zhong2021phase, Xanadu2022nature, deng2023gaussian} and RCS~\cite{Boixo2018nature, arute2019quantum, wu2021strong, morvan2023phase}.

However, it is still unclear whether the quantum advantage can be maintained for larger quantum systems, and beyond that, for the asymptotic limits as system size scales. 
The major obstacle to the scalability of the quantum advantage is the uncorrected noise on near-term quantum devices, which suppresses the quantumness and may eventually allow efficient classical simulation. 
In particular, as the noise rate typically grows with circuit depth, the circuit depth should not be too large to accomplish the hardness of classical simulation for noisy devices.
Otherwise, the accumulated noise induced by a large-depth circuit may result in the loss of quantumness and classical intractability.
On the other hand, the circuit depth should also be large enough to achieve the simulation hardness, as a shallow depth circuit leads to the weakly entangled quantum state, which is often efficiently simulable by classical algorithms~\cite{Vidal2003prl, Vidal2004prl, Verstraete2004prl, Napp2022prx, Deshpande2018prl, Qi2022PRA, Oh2022PRL}. 
Therefore, to attain the scalable quantum advantage with noisy quantum devices, it is necessary to identify an appropriate regime, where the depth is large enough to generate the quantumness and not too large to lose the classical intractability by the uncorrected noise.

To be more specific about the effect of noise on the complexity, there are many results proposing the possibilities of the efficient classical simulation of noisy BS and RCS as system size scales, hindering the scalability of the quantum advantage.
For the RCS case, Refs.~\cite{Aharonov1996arxiv, Boixo2017arxiv, Gao2018efficient, Wang2021PRXQ, Deshpande2022prx} suggested a probability distribution of most quantum circuits of super-logarithmic depth with a constant level of noise (depolarizing, Pauli, etc) applied for each unit depth converges to a uniform distribution. 
Similarly, for the BS case, there have been many results~\cite{kalai2014gaussian, Renema2018arxiv, Renema2018prl, Shchesnovich2019pra, Moylett2020QST, Brod2020quantum, Oh2023noisy, Oh2023Tensor} about efficient classical algorithms to simulate noisy BS, with various noise models, such as photon loss and partial distinguishability of photons. 
Those results suggest that noisy BS can be efficiently simulated if the noise rate is sufficiently large.

As an example, for photon loss, if the output photon numbers are less than $\mathcal{O}(\sqrt{N})$ for input photon number $N$, BS becomes efficiently simulatable within a constant total variation distance~\cite{Oszmaniec2018njp, Patron2019quantum, Qi2020prl}.
Since the output photon number is given by $N_{\text{out}}=NT^{D}$ for circuit depth $D$ and a transmission rate $T$ applied at each depth, it implies the easiness of classical simulation for superlogarithmic circuit depth unless the transmission rate per depth $T$ increases with system size, which is typically unrealistic.
Therefore, according to the above discussions, circuit depth superlogarithmic with system size would possibly rule out the scalable demonstration of the quantum advantage \cite{Patron2019quantum}.

\iffalse
Moreover, some results claim that optical circuits with polynomial depth \co{How about just superlog? Seems like the paragraph eventually says superlog.} will be already noisy enough for efficient simulation of BS in the asymptotic regime.
To make this point clear, we take photon loss as an example of a noise model, which is indeed a dominant source of errors in optical systems. Refs.~\cite{Oszmaniec2018njp, Patron2019quantum, Qi2020prl} proved that if the output photon numbers are less than $\mathcal{O}(\sqrt{N})$ for input photon number $N$, BS becomes efficiently simulatable within at most a constant total variation distance. 
Meanwhile, for a transmission rate $T$ applied at each depth, the output photon number is given by $N_{\text{out}}=NT^{D}$ for circuit depth $D$. 
This implies the easiness of classical simulation for superlogarithmic circuit depth unless the transmission rate per depth $T$ increases with system size, which is typically unrealistic.
Therefore, according to the above discussions, circuit depth superlogarithmic with system size would possibly rule out the scalable demonstration of the quantum advantage.
\fi

The above example clearly shows that as circuit depth grows, more noises are accumulated in the system.
Hence, an obvious way to suppress the effect of loss is to consider shallow-depth circuits and investigate if we can still maintain hardness in shallow-depth circuits.
A crucial factor in proving the hardness of approximate sampling, specifically additive-error sampling, using the state-of-the-art proof technique is the average-case \#P-hardness of output probability approximation within a certain additive imprecision (viz., on average over possible outcomes) \cite{aaronson2011computational, Bouland2019nature, hamilton2017gaussian, deshpande2021quantum, movassagh2020quantumsupremacy, kondo2021quantumsupremacy, Bouland2022noise}. 
%One evidence often cited as the average-case hardness, but not sufficient, is an anti-concentration property~\cite{aaronson2011computational, Bouland2019nature, Bremner2016prl, Hangleiter2018anticoncentration, Bermejo2018prx, Dalzell2022prxq}, which is a measure of the flatness of probability distribution.
Here, anti-concentration property~\cite{aaronson2011computational, Bouland2019nature, Bremner2016prl, Hangleiter2018anticoncentration, Bermejo2018prx, Dalzell2022prxq}, which is a measure of the flatness of probability distribution, comes into the proof of the average-case hardness up to the additive error. 
For the RCS case, recent results found that the probability distribution of logarithm depth random quantum circuit is anti-concentrated~\cite{barak2020spoofing, Dalzell2022prxq}, but poorly anti-concentrated for sub-logarithm depth~\cite{Deshpande2022prx}. 
Those results suggested hope that log-depth RCS may have the average-case hardness property and thus may be a suitable candidate for realizing scalable quantum advantage even with noisy devices.
However, more recently, Ref.~\cite{Aharonov2022noisy} presented a polynomial-time classical simulation of noisy log-depth RCS with a constant level of noise per depth using the anti-concentration property, deflating the hope.

Meanwhile, for the BS case, output probability distributions and their average-case approximation hardness for shallow-depth BS have been less studied compared to the RCS case, to the best of our knowledge.
There have been some results about efficient classical algorithms to simulate shallow-depth BS, but they can be applied only in particular circumstances.
Specifically, Ref.~\cite{Qi2022PRA} suggested a classical algorithm that can simulate 1-dimensional logarithm depth BS efficiently. 
Also, Refs.~\cite{Deshpande2018prl, Oh2022PRL} proposed classical algorithms which can efficiently simulate BS for the higher dimensional local circuits but on condition that input sources are well-separated, under a certain depth such that the well-separated input sources do not correlate much with each other. 
Although those algorithms can clarify the easiness of simulating shallow-depth BS for specific cases, it remains an open question whether the simulation hardness of shallow-depth BS exists with more general setups, e.g., allowing different circuit architectures, circuit ensembles, and input configurations. 
Hence, the motivation of our work is to investigate the possibility of shallow-depth BS that achieves simulation hardness, by analyzing the behavior of output probability distributions in low-depth regimes for general setups.

In this work, we find that the probability distribution is too concentrated to achieve the average-case hardness of probability approximation for local linear-optical circuits under a certain polynomial depth, regardless of circuit ensembles and input configurations, both for Fock-state BS (FBS) and Gaussian BS (GBS) schemes. 
We prove that from the structure of the local parallel circuit architecture, a typical way to build a random linear-optical network using geometrically local gates, most of the outcomes are forbidden at the shallow depth regime, i.e., have zero probability, regardless of circuit ensembles. 
Besides, if we employ local {\it random} circuit ensembles, their diffusive properties make probability distribution concentrated even inside the permitted outcomes, requiring additional circuit depth to get out of the concentrated regime.

Following the above examination, we propose a linear-optical circuit architecture that can resolve the above issues within logarithm depth, using geometrically non-local gates.
We numerically examine that the corresponding circuit architecture with each gate drawn from the local random unitary can well imitate the behavior of the global Haar random unitary, the requirement to achieve evidence of the average-case hardness, in a shallow depth regime.
Specifically, output probability distribution and entanglement generation of the above random circuit setup show fast convergence toward those of the global Haar random circuit with increasing circuit depth.
Those results highlight that the corresponding circuit architecture shows a potential to achieve the average-case hardness in the shallow depth regime and be utilized as an architecture for scalable quantum advantage with BS.

The paper is organized as follows.
In Sec.~\ref{averagecasehardnessintroduction}, we begin with a brief introduction of the average-case hardness of approximating output probabilities of BS and discuss the possible issues that may arise when we lower the depth.
In Sec.~\ref{limitationsoflocal}, we present problems of geometrically local architectures, which hinder the anti-concentration and thus do not allow the average-case hardness below a certain polynomial depth. 
In Sec.~\ref{nonlocalhypercubicstructure}, we propose a geometrically non-local circuit architecture that can resolve the problems we addressed within logarithm depth.
We numerically examine the probability distribution and entanglement of the circuit for various system sizes.   
In Sec.~\ref{discussion}, we conclude with a few remarks.

\section{Average-case hardness and anti-concentration}\label{averagecasehardnessintroduction}

We first briefly recall the average-case hardness of BS and investigate if a similar hardness still holds in the low-depth regime, examining potential issues that may arise when we lower the depth. 
%\co{What is the meaning of current statement?}
There are two major schemes of BS, which are FBS and GBS, and we consider the average-case hardness of FBS first. 
The output probability of FBS is expressed as permanent, which is worst-case \#P-hard to compute \cite{aaronson2011computational, valiant1979complexity}.
To achieve simulation hardness of additive-error sampling, whose error is bounded by total variation distance with its size at least inverse polynomial order, additively approximating most of the output probabilities of the sampler should be \#P-hard (i.e., average-case \#P-hardness). 
This comes from the fact that a sampler with total variation distance error can have a large additive error for a single output probability, but still have a small additive error on average over all possible outcomes due to Markov's inequality.
Here, to get the average-case hardness of the output probability approximation, the implementation of random circuit instances is required.
Specifically, the current proof technique requires that unitary matrices corresponding to the random circuits must be drawn from Haar measure on U($M$), where $M$ corresponds to the total mode number.  
Under some plausible assumptions (see \cite{aaronson2011computational} for more details), approximating most of the output probabilities of FBS on average over random linear-optical circuits is \#P-hard, which can be represented as follows.
\begin{conjecture}\label{Averagehard}
For mode number $M$ and photon number $N$, approximating most output probabilities of most linear-optical circuits within additive error $\rm{poly}\it{(N)^{-1}\frac{N!}{M^N}}$ is \#P-hard. 
\end{conjecture}
The above problem is solvable within a finite level of polynomial hierarchy using the additive-error sampler, which implies that efficient classical simulation of the corresponding sampler implies polynomial hierarchy collapses to the finite level, which is highly unlikely. 

Also, a similar argument has been established for GBS in~\cite{hamilton2017gaussian, deshpande2021quantum}, whose output probability is now expressed as hafnian, which is at least as hard as permanent to compute.
Those results suggest that for fixed output photon number $N$, approximating most of the output probabilities of GBS on average over random linear-optical circuits is \#P-hard, under plausible assumptions. 
Accordingly, implementing $M$-mode random linear optical circuits corresponding to $M$-dimensional Haar random unitary matrices is necessary to find evidence of the average-case hardness of BS using the current techniques (both for FBS and GBS schemes).
However, the exact implementation of Haar random unitary matrices requires at least polynomially large circuit depth (e.g., results in Ref.~\cite{Russell2017njp}).
%\co{Just to confirm, do we write precisely? or more precisely? I feel like people use only `more precisely'. I'm not sure.}
More precisely, as the space of $M$ dimensional unitary matrices is determined by $M^2$ independent parameters~\cite{Zyczkowski1994jpa}, $\Omega(M^2)$ number of local random unitary gates are required to implement global Haar random unitary, which can only be accomplished by $\Omega(M)$ circuit depth. 
%\cor{In other words, low-depth random circuits, specifically random circuits with depth smaller than $\Theta(M)$, cannot achieve the exact $M$ dimensional Haar measure and thus cannot be expected to have the average-case hardness based on current techniques, unless they can well approximate the output probability distribution of $M$ dimensional Haar random circuit.}
%Therefore, to show the average-case hardness of BS in the low-depth regime, we need to find a different approach. 

One way to find evidence of the average-case hardness of approximating output probabilities in the low-depth regime, as we introduced in the previous section, is by investigating the anti-concentration property of the corresponding output probability distribution.
More formally, for the output probability of boson sampling $p_{\bm{s}}$ and for $\xi > 0$, the anti-concentration property can be expressed as 
\begin{equation}
    \Pr_{U}\left[ p_{\bm{s}} < \frac{\text{poly}(N,1/\xi)^{-1}}{\binom{M+N-1}{N}} \right] < \xi ,
\end{equation}
where the probability is over the unitary matrix $U$ drawn from Haar measure on U($M$) that corresponds to the circuits. Note that because of the symmetry from the Haar measure, the choice of the outcome $\bm{s}$ may be arbitrary.
Anti-concentration plays a key role in proving the current average-case hardness of various circuit families, and accordingly, there have been many efforts to show whether there exists anti-concentration for various circumstances~\cite{aaronson2011computational, Bremner2016prl, morimae2017hardness, Hangleiter2018anticoncentration, deshpande2021quantum, Dalzell2022prxq, Deshpande2022prx,fefferman2023effect}.
Although anti-concentration itself does not guarantee the classical hardness of sampling, it has been shown that classical simulation of certain circuit families becomes easy if the output probability distribution is sparse enough (namely, overly concentrated)~\cite{schwarz2013simulating, roga2020classical, bravyi2023classical}.
%Those results stress that anti-concentration is necessary to get the simulation hardness. 
%This shows that concentrated output distribution is vulnerable to efficient classical simulation for such cases, which stresses that anti-concentration is necessary to achieve simulation hardness in such cases.
Those results indicate that concentrated output probability distribution makes classical simulation easier for such cases, demonstrating that anti-concentration is a crucial step in establishing the simulation hardness.
%However, if the output probability distribution is sparse enough (i.e., overly concentrated), it has been shown that classical simulation of certain circuit families becomes easy~\cite{schwarz2013simulating, roga2020classical, bravyi2023classical}. 
%Accordingly, there have been many efforts to show whether there exists anti-concentration for various circumstances~\cite{Bremner2016prl, morimae2017hardness, Hangleiter2018anticoncentration, Dalzell2022prxq, Deshpande2022prx, fefferman2023effect}.

Moreover, anti-concentration proposes the minimum depth required for the current average-case hardness proof, as the lack of anti-concentration directly ruins the proof technique even though it does not necessarily rule out the classical hardness of boson sampling.
%i.e., necessary but not sufficient condition. }
%For example, if many probabilities are very small, those probabilities are not assured to be hard to estimate within additive error, as even a small additive error would be relatively large to those probabilities and thus possibly ruin the average-case hardness.
%\co{Why suddenly multiplicative error?}
Specifically, when the output probability distribution is concentrated such that most of the output probabilities are very close to zero and easy to approximate within additive error by a trivial algorithm that always outputs the value ``0", there is no average-case hardness of approximating output probabilities.
Hence, if the output probability of BS below a certain depth is poorly anti-concentrated as described above, the corresponding depth suggest a prerequisite for the average-case hardness of output probability approximation, as we cannot achieve the average-case hardness below the depth.

\begin{figure}[t]
%\centering
\includegraphics[width=\textwidth]{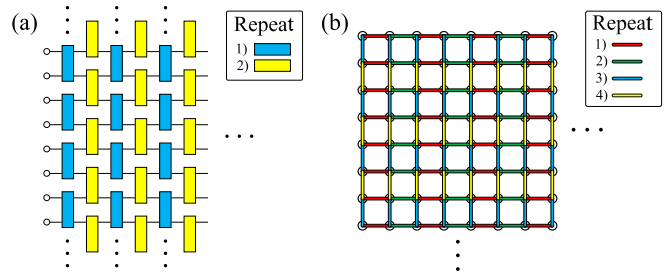}
\caption{Schematics of (a) 1-dimensional and (b) 2-dimensional local parallel circuit. Unitary gates with the same colors are applied in parallel, and the sequence is $2$ steps of the parallel application of local gates for each dimension.
Here, such application of gates with the same color implies the unit depth. 
%For (b), gates corresponding to red, green, blue, yellow colors are applied parallel for each sequence. 
}
\label{fig:parallel circuit}
\end{figure}

\section{Limitations of geometrically local architectures}\label{limitationsoflocal}

%\co{Shall we change the section structure? Because now the structure is too simple to separate different topics. Typically after intro, everything is results. So we do not need to specify this is result part (unless we submit it to some nature journals where they require us to use result section).}
We consider an $M$-mode linear optical circuit characterized by unitary operator $\hat{\mathcal{U}}$. After the evolution, input mode operators of the system are linearly transformed by a unitary matrix $U$ defined as
\begin{align}
\hat{a}_i \Longrightarrow \hat{\mathcal{U}}^{\dag}\hat{a}_i\hat{\mathcal{U}} = \sum_{j=1}^{M}U_{ij}\hat{a}_{j}
\end{align}
The probability of each outcome of BS is represented by a function of the unitary matrix $U$, which corresponds to permanent and hafnian for FBS and GBS, respectively.

A typical way to construct an $M$-mode random linear-optical network is a \textit{local parallel} circuit architecture, a parallel array of geometrically local beam splitters that can interact with only nearest-neighbor modes \cite{Clements2016optica}.
%a parallel array of geometrically local beam splitters that can interact with only nearest-neighbor modes, a generalization of 
%We can generalize the architecture by raising its dimension, so 
%Generally, we define a $d$-dimensional local parallel circuit as a parallel array of local beam splitters (where the ensemble they follow is not specified yet) for each dimension. 
The corresponding circuit architecture has often been used both for theoretical~\cite{Zhang2021npjq, Oh2022PRL, Russell2017njp} and experimental~\cite{Tang2018sciadv, zhong2020quantum, zhong2021phase, Tang2022prl} studies, as a building block for random linear-optical circuits.
There are schematics of the 1-dimensional and 2-dimensional local parallel circuits in Fig.~\ref{fig:parallel circuit}, and we can easily generalize those circuits to a $d$-dimensional local parallel circuit. 
More formally, the $d$-dimensional local parallel circuit with depth $D$ consists of $D/2d$ rounds, where a single round consists of $2d$ steps, $2$ steps of the parallel application of local gates for each dimension.
Here and throughout, we emphasize that a unit depth in our work refers to a single step of parallel application of gates at a fixed dimension, to avoid any confusion.
%\cor{To avoid confusion, we stress that we set a unit depth as a single step of parallel application of gates at a fixed dimension.}
%which means that the application of gates with the same color in Fig.~\ref{fig:parallel circuit} implies the unit depth. 

%Such unitary evolution $\hat{\mathcal{U}}$ and corresponding unitary matrix $U$ are determined by linear optical circuits. 
With a given linear optical circuit architecture, circuit lightcones can be determined by the gates composing the circuit. 
%Here, we use a similar definition of the circuit lightcone from \cite{Deshpande2022prx}. \co{Maybe this is typical definition, so we do not need to emphasize that the notation is borrowed from the literature.} 
We define $L_D(i)$ as a \textit{forward} lightcone, which is a set of all modes at depth $D$ connected with input mode $i$ via the gates of the architecture. Similarly, \textit{backward} lightcone $L_D^{\dag}(i)$ is a set of all input modes connected with mode $i$ at depth $D$ along the gates. 
Intuitively, the lightcone indicates a set of modes photons can spread from a mode via the gates within a finite depth.

\begin{figure*}[bt!]
\includegraphics[width=0.9\textwidth]{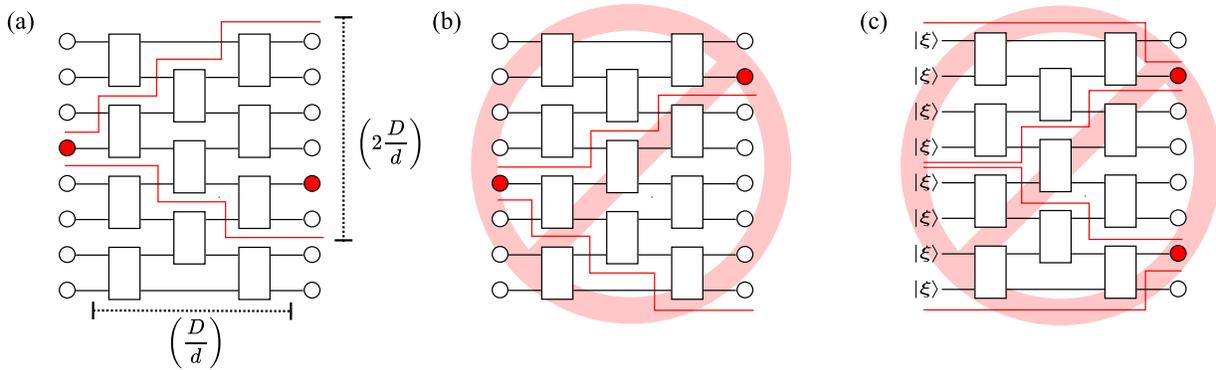}
\caption{(a) A schematic of the forward lightcone of a 1-dimensional parallel circuit. We employ a paradigmatic scheme such that the left side denotes the input and the right side denotes the output. From the parallel architecture, the size of the lightcone for each dimension is twice of depth for each dimension. (b) An example of a forbidden outcome for FBS, which does not satisfy the constraint \eqref{constraint1}: input and output photons are not connected by the lightcone. (c) An example of a forbidden outcome for GBS, which does not satisfy the constraint \eqref{constraint2}: a pair of output photons must be originated from a single source.}
\label{fig:lightcones}
\end{figure*}

To specify the circuit ensemble for random circuit instances, we also define the $d$-dimensional local parallel \textit{random} circuit, which is $d$-dimensional local parallel circuit with each gate drawn from the Haar measure on U($2$) independently.
%similar to the above, except that the circuit now follows a specific ensemble, viz., each gate is drawn from the Haar measure on U($2$) independently.
%\cor{Note that this circuit corresponds to the linear-optical version of random circuit setups of RCS, so it is not an absurd way to manipulate a random ensemble for given circuit architecture.}
Such configurations are motivated by recent experimental setups \cite{zhong2020quantum, zhong2021phase, Xanadu2022nature, deng2023gaussian}. 
For the following sections, we consider the $d$-dimensional local parallel (random) circuit as default and prove that the probability distribution is poorly anti-concentrated under a certain polynomial circuit depth, both for FBS and GBS schemes, implying the lack of the average-case hardness for shallow depth circuits.
%Additionally, to avoid confusion, we define a unit depth as a single step of parallel application of gates at a fixed dimension, so a single round of $d$-dimensional local parallel circuit consists of $2d$ depth.

\subsection{Fock-state boson sampling}
%\co{If we are going to say FBS instead of Fock-state BS, how about we start to do this from the introduction? same for the GBS.}
For FBS, we set the basic parameters as follows. The system is composed of total $M$ modes, and the input state is composed of $N$ single-photon states and vacuum states for the rest modes where $M$ and $N$ are polynomially related by $M = c_0N^\gamma$ with $\gamma \ge 1$. 
We use the notation $\bm{t}$ as the position of input single photons, where $\bm{t}\in \mathbb{Z}_{\geq 0}^{N}$ with $1\leq t_i\leq M$ and $t_1< t_2<\cdots< t_N$ such that each element $t_i$ denotes a mode where a photon is injected. 
Also, $\bm{s}$ denotes the output photons' position, where $\bm{s}\in \mathbb{Z}_{\geq 0}^{N}$ with $1\leq s_i\leq M$ and $s_1\le s_2\le\cdots\le s_N$ (equality denotes collision outcomes), so each element $s_i$ corresponds to a mode where a photon is measured.
Here, the number of possible patterns of $\bm{s}$ is given by $\binom{M+N-1}{N}$. 

Using those notations, the output probability of FBS to obtain the outcome $\bm{s}$ is expressed as \cite{aaronson2011computational}
\begin{equation}\label{bsprobability}
p_{\bm{s}} = \frac{\left|\text{Per}\left( U_{\bm{s},\bm{t}}\right)\right|^2}{\bm{s}!} = \frac{1}{\bm{s}!}\left|\sum_{\sigma\in S_N}\prod_{j=1}^{N} U_{s_{\sigma(j)},t_{j}} \right|^2,
\end{equation}
where $U_{\bm{s},\bm{t}}$ is a submatrix of circuit unitary matrix $U$ such that $(U_{\bm{s},\bm{t}})_{i,j} = U_{s_i, t_j}$.
Also, $\bm{s}!$ denotes a product of the multiplicity of all possible values in $\bm{s}$, and $\sigma \in S_N$ denotes permutations along $N$ modes. 
%\co{Could you write the definition of $U_{\bm{s},\bm{t}}$?}
%\co{Also, here, isn't $N$ modes correct only when we consider collision-free outcomes?}

As mode interactions cannot occur outside the circuit lightcone, we have $U_{ij} = 0$ for $i \notin L_D(j)$. 
Hence, if the size of the lightcone is small such that for given $\bm{s}$, $s_{\sigma(j)} \notin L_D(t_{j})$ for all $\sigma \in S_N$ and at least one $j \in [N]$, this implies that the given outcome $\bm{s}$ is unphysical, suggesting $p_{\bm{s}} = 0$.
We can check a simple example of a forbidden outcome of FBS in Fig.~\ref{fig:lightcones}(b). 
Intuitively, the number of unphysical outcomes would increase if the lightcone size gets smaller, as photon propagation along different modes becomes more restricted.

From the structure of the local parallel architecture, we find that the size of the lightcone $ \left|L_D(i)\right|$ grows as
\begin{equation}\label{upperboundoflightcone1}
 \left|L_D(i)\right| \le \left(\frac{2D}{d}\right)^d,
\end{equation}
where $D$ denotes circuit depth and $d$ denotes circuit dimension. 
The right-hand side of Eq.~\eqref{upperboundoflightcone1} comes from the fact that the size of the lightcone is $2D/d$ for each dimension as we defined a unit depth as a single step of parallel application of gates (see Fig.~\ref{fig:lightcones}(a)). Equality holds when the lightcone does not meet geometrical boundaries.

As we can see, for local parallel architecture, lightcones grow polynomially with depth by Eq.~(\ref{upperboundoflightcone1}). 
Hence, a large portion of outcomes might be unphysical and thus have zero output probability for the low-depth regime (e.g., logarithm depth), as photons are localized inside lightcones of each input mode, whose sizes are very small in this regime. 
Following the above intuition, we first prove that only an exponentially small portion of the outcomes of FBS is allowed under a certain polynomial depth, regardless of input configurations and ensembles of gates composing the circuit. 
%\co{How about we simply say the probabilities are zero, instead of forbidden when it needs to be formal?}
\begin{theorem}\label{theorem1}
(Any circuit ensemble) For $d$-dimensional local parallel circuit of depth $D \le \kappa_0N^{\frac{\gamma - 1}{d}}$ and arbitrarily chosen input mode configuration,  most of the outcomes of FBS have zero output probability for a constant $\kappa_0 = e^{\frac{1}{d}}c_{0}^{\frac{1}{d}}d/{2}$.
\end{theorem}
%with a constant $\kappa_0 = e^{\frac{1}{d}}c_{0}^{\frac{1}{d}}d/{2}$,
\begin{proof}
From Eq.~\eqref{bsprobability}, the probability is non-zero only if at least one path exists between each input and output mode within each lightcone, and this statement can be written as
\begin{equation}\label{constraint1}
    \exists \sigma \in S_N \; \forall j \in [N] \; \text{~such that~}\; s_{\sigma(j)} \in L_D(t_{j}).
\end{equation}
%, outcomes are forbidden if for each $\sigma \in S_N$ there exists an index $j \in [N]$ such that $U_{s_{\sigma(j)},t_{j}} = 0$ which is equivalent to $s_{\sigma(j)} \notin L_D(t_{j})$. In other words, 
%\co{This is in the proof. So, I think if you want to give intuition, it would be better before presenting lemma, so that the proof becomes more compact and concise. Indeed, we'd better explain some physical intuition before presenting results by mathematical lemma and theorem.}
Let $\mathcal{M}$ be the number of outcomes $\bm{s}$ satisfying the constraint (\ref{constraint1}). Then $\mathcal{M}$ is bounded by
\begin{equation}\label{upperboundofN}
\mathcal{M} \le \prod_{i=1}^{N}\left|L_D(t_i)\right|.
\end{equation}
The right-hand side of Eq.~(\ref{upperboundofN}) counts the outcome by one-to-one correspondence from the input to the output, so at least one output photon is present at each input lightcone, satisfying the constraint (\ref{constraint1}). 
Clearly, it is an upper bound for any input configuration $\bm{t}$, because it double-counts the output photon distribution inside the overlapping region of lightcones as photons in the region are indistinguishable. 
In addition, the inequality at Eq.~(\ref{upperboundofN}) becomes tighter as the more $L_D(t_i)$ do not overlap each other, which corresponds to input configurations that photon sources are far from each other (e.g., setups from Refs.~\cite{Deshpande2018prl, Oh2022PRL}).

Here we note that $\mathcal{M}$ can be exponentially large in terms of system size, but still, it can be much smaller than all possible output patterns (i.e., $\binom{M+N-1}{N}$).
The ratio of permitted outcomes over all possible outcomes, which we will denote as $\Delta$, is bounded by
\begin{align}
    \Delta &= \frac{\mathcal{M}}{\binom{M+N-1}{N}} \\
&\le \sqrt{2\pi N}e^{\frac{1}{12N}}\left(\frac{2^dD^dN}{ed^dM}\right)^N \\
&< 3\sqrt{N}\left(\frac{2^dD^dN^{1-\gamma}}{ed^dc_0}\right)^N,
\end{align}
where we used Stirling's formula and the fact that $N$ is larger than the unity. 
Here, we only employed the structure of the local parallel architecture (viz., maximal lightcone size), and did not take into account the ensemble the circuit follows. 
Therefore, the result shows that for depth $D \le \kappa_0N^{\frac{\gamma - 1}{d}}$ with a constant $\kappa_0 = e^{\frac{1}{d}}c_{0}^{\frac{1}{d}}d/{2}$, asymptotically $\Delta$ is exponentially small with system size, regardless of input configuration $\bm{t}$ and circuit ensemble. 
\end{proof}
%&< 3\sqrt{N}e^{-N[(\gamma - 1)\log N - d\log (\frac{2D}{e^{\frac{1}{d}}c_{0}^{\frac{1}{d}}d} ) ]},

Theorem \ref{theorem1} implies that for circuits with a depth below a certain threshold, most output instances have zero probabilities for {\it any} circuits. 
Therefore, most of the output probabilities are easy to approximate for any circuit instance, which undermines the average-case hardness of the output probability approximation for any circuit ensemble. 

Furthermore, many current experiments compose their circuits by choosing local beam splitters randomly, i.e., the local parallel random circuit we introduced.
In that case, we find that most photons propagate in a region smaller than the actual lightcone, requiring additional circuit depth to get out of the easiness regime.

%Furthermore, a bigger issue arises when we introduce random circuit instances, which are necessary for the average-case hardness, 
%because for such cases most of the photon dynamics occur in a region smaller than the lightcone, requiring additional circuit depth to get out of the easiness regime. This argument can be confirmed by the following theorem. 
\begin{theorem}\label{theorem2}
(Local random ensemble) For $d$-dimensional local parallel random circuit of depth $D\le \alpha_0 N^{\frac{2(\gamma-1)}{d}-\lambda}$ with a constant $\alpha_0 = e^{\frac{2}{d}}c_{0}^{\frac{2}{d}}\beta d/{2}$, for any $\lambda > 0$, $0 < \beta < 1$, and input mode configuration, it is easy to estimate the output probabilities of FBS within additive error $\epsilon = \rm{poly}\it{(N)^{-1}\frac{N!}{M^N}}$, for $1-\xi$ portion of output instances with probability $1-\delta$ over the random circuit instances, where $\xi$ and $\delta$ are exponentially small with system size.
\end{theorem}

% approximating probabilities of FBS over $1-\xi$ portion of outcomes within additive error is easy with probability larger than $1-\delta$ over the random instance, where $\xi$ and $\delta$ are exponentially small with system size.

%The theorem states that we can approximate most output probabilities of most of the linear optical circuits under a certain depth. 
The term `easy' means that we can approximate the probability well within allowed additive error $\epsilon$ by a trivial algorithm that outputs the value ``0".
%whose size is $\poly(N^{-1})|\mathcal{S}|^{-1}$, where $|\mathcal{S}|$ denotes the size of the sample space which is $\binom{M+N-1}{N}$ for our FBS case. 
% 0으로 근사해도 잘 맞는다. 
%We prove that before photons effectively spread far enough, the probability distribution is concentrated to a small portion of the outcomes, so the probabilities of the rest are easy to estimate. 
%\co{The following paragraph seems redundant. Can you combine this from the previous paragraphs?
%I guess here we have three different regimes. First, any circuits, an ensemble of circuits, and local random circuits. The separation needs to be clearer.}
%Note that Theorem \ref{theorem1} is derived only by the architecture of the parallel circuit which characterizes the lightcone. 
%For random circuit instances, we consider the parallel circuit composed of geometrically local random beam splitters, each corresponds to a random matrix drawn from Haar measure on U(2). 
From the diffusive properties of the local parallel random circuit studied by Ref.~\cite{Oh2022PRL, Zhang2021npjq}, we find that the probability distribution is still poorly anti-concentrated below a certain depth which is larger than $\kappa_0N^{\frac{\gamma - 1}{d}}$.

%\co{How about we move this proof to Appendix and provide a proof sketch in the main text? The proof of lemma 1 is very relevant to our claim but this proof is not directly related to our main claim.}
%\co{proof sketch?}
%\proofsketch 
To sketch the proof of Theorem \ref{theorem2}, we define the `effective' lightcone such that truncating the outcomes that at least one photon propagates outside this lightcone results in exponentially small total variation distance from the original distribution, below a certain polynomial depth. 
In this regime, most of the probability instances are easy to estimate, including the unphysical outcomes and the truncated outcomes (though not unphysical) outside the effective lightcones. 

\begin{proof}
See Appendix \ref{theorem2proof}.
\end{proof}

\subsection{Gaussian boson sampling}
In this section, we consider an $M$-mode GBS with a product of $K$ single-mode squeezed vacuum (SMSV) states with equal squeezing and vacuum states for the rest modes. 
We focus on the output photon number $2n$, i.e., there are $n$ pairs of output photons. 
Now $M$ and $n$ are polynomially related by $M = c_1n^\gamma$ where $\gamma \ge 1$, and let $K$ be any number satisfying $n\le K\le M$. 
For convenience, we adjust the squeezing parameter such that the mean photon number for our setup is also $2n$, which implies $2n = K\sinh^2{r}$.
Similarly from the FBS case, we use the notation $\bm{t}$ as the input vector, where $\bm{t}\in \mathbb{Z}_{\geq 0}^{K}$ with $1\leq t_i\leq M$ and $t_1< t_2<\cdots< t_K$ such that each element $t_i$ denotes a mode where SMSV is injected. Also, $\bm{s}$ denotes an output vector, where $\bm{s}\in \mathbb{Z}_{\geq 0}^{2n}$ with $1\leq s_i\leq M$ and $s_1\le s_2\le\cdots\le s_{2n}$. 

Using the above conventions, the probability of obtaining outcome $\bm{s}$ for GBS is proportional to
\begin{align}
&p_{\bm{s}} \propto \left|\text{Haf}\left( \left(U {I}_K U^T\right)_{\bm{s}}\right)\right|^2 \nonumber \\&= \left|\sum_{\mu\in\text{PMP}}\prod_{j=1}^{n}\left(\sum_{k=1}^{M}\sum_{l=1}^M U_{s_{\mu(2j-1)}k}({I}_K)_{k,l} U_{s_{\mu(2j)}l}\right) \right|^2,
\end{align}
where ${I}_K$ denotes the projection matrix to the input $K$ modes, and PMP denotes all possible perfect matching permutations along $2n$ modes.

Similar to the FBS case, unphysical outcomes are naturally forbidden as $p_{\bm{s}}=0$.
More precisely, If the size of the lightcone is small such that for given $\bm{s}$, $\sum_{k}\sum_{l} U_{s_{\mu(2j-1)}k}({I}_K)_{k,l}U_{s_{\mu(2j)}l} = 0$ for all $\mu \in \text{PMP}$ and at least one $j\in[n]$, then the outcome $\bm{s}$ is forbidden. 
This phenomenon comes from the fact that the input state is a product of SMSV which is a superposition of even number Fock-states, i.e., photons are always generated as a pair. 
Therefore, if there exists an output photon such that its backward lightcone does not share the input source with backward lightcones of the other photons, the corresponding outcome is unphysical and thus prohibited.
We can check an example of a forbidden outcome of GBS in Fig.~\ref{fig:lightcones} (c).

Following the intuition, we prove that only an exponentially small portion of the outcomes of GBS is allowed under a certain polynomial depth, regardless of input configurations and ensembles of gates composing the circuit.
%\co{Again, how about we begin with an intuitive explanation before Lemma?}
\begin{theorem}\label{theorem3}
(Any circuit ensemble) For $d$-dimensional local parallel circuit of depth $D \le \kappa_1n^{\frac{\gamma - 1}{d}}$, arbitrary $K$ within $n\le K\le M$ and arbitrarily chosen input mode configuration, most of the outcomes of GBS have zero output probability for a constant $\kappa_1= c_1^{\frac{1}{d}}e^{\frac{1}{d}}d/2^{\frac{1}{d}+2}$. 
\end{theorem}
% with a constant $\kappa_1= c_1^{\frac{1}{d}}e^{\frac{1}{d}}d/2^{\frac{1}{d}+2}$,
\begin{proof}
The probability of GBS is non-zero only if $2n$ output photons are paired up and share the input source inside their backward lightcones, and this statement can be written as
\begin{align}\label{constraint2}
    \exists \mu  \; \forall j \; \exists i \; \text{~such that~} \; t_i \in  L_D^{\dag}(s_{\mu(2j-1)}) \cap L_D^{\dag}(s_{\mu(2j)})  .
\end{align}
Since we want to upper-bound the number of permitted outcomes, we alleviate the constraint \eqref{constraint2} by taking $K = M$, which denotes the full-mode input. Then, from \eqref{constraint2}, the constraint for the index $i$ vanishes and the last term reduces to $L_D^{\dag}(s_{\mu(2j-1)}) \cap L_D^{\dag}(s_{\mu(2j)}) \neq \emptyset$ . Then we can approximately count the number of possible outcomes as follows: $(i)$ Pick $n$ pairs over $M$ modes with replacement. $(ii)$ Count all possible alignment for each $n$ pair connected by lightcones. 

Let $\mathcal{M}$ be the number of outcomes $\bm{s}$ satisfying the constraint \eqref{constraint2}. Then $\mathcal{M}$ is bounded by
\begin{equation}\label{numberofpossibleoutcome2}
\mathcal{M} \le \frac{1}{2^n}\sum_{\textbf{r}}\prod_{i=1}^{n}\left|L_D \circ L_D^{\dag}(r_i)\right|,
\end{equation}
where $\bm{r}\in \mathbb{Z}_{\geq 0}^{n}$ with $1\leq r_i\leq M$ and $r_1\le r_2\le\cdots\le r_n$ denotes the possible mode selection at step $(i)$ so the number of possible alignment of $\bm{r}$ is $\binom{M+n-1}{n} $. The factor $\frac{1}{2^n}$ comes from the indistinguishability of the photons between step $(i)$ and $(ii)$. As we allowed double counting at overlapping regions and weakened the constraint \eqref{constraint2} by taking $K = M$, the right-hand side of Eq.~\eqref{numberofpossibleoutcome2} certainly implies an upper bound.

From the structure of $d$-dimensional parallel circuit architecture, with generalization of Eq.~\eqref{upperboundoflightcone1}, $|L_D \circ L_D^{\dag}(r_i)|$ has the corresponding bound
\begin{equation}\label{upperboundoflightcone2}
\left|L_D \circ L_D^{\dag}(r_i)\right| \le \left(\frac{4D}{d}\right)^d,
\end{equation}
where equality holds when lightcones do not meet geometrical boundaries. Now $\Delta$ is the ratio of outcomes satisfying \eqref{constraint2} over all possible output distributions, and it is bounded by 
\begin{align}
\Delta &= \frac{\mathcal{M}}{\binom{M + 2n - 1}{2n}} \\
&\le \frac{1}{2^n}\left[4^d\left(\frac{D}{d}\right)^d\right]^n\frac{\binom{M+n-1}{n}}{\binom{M + 2n - 1}{2n}} \\
&\le \sqrt{2}e^{\frac{1}{24n}}\left( \frac{2^{2d+1}D^d}{ed^d}\frac{n}{M} \right)^n \\
&< 2\left( \frac{2^{2d+1}}{ed^dc_1}D^dn^{1-\gamma} \right)^n.
\end{align} 
We used Stirling's formula and the fact that $n$ is larger than the unity. 
This result shows that for depth $D \le \kappa_1n^{\frac{\gamma - 1}{d}}$ with $\kappa_1 = c_1^{\frac{1}{d}}e^{\frac{1}{d}}d/2^{\frac{1}{d}+2}$, $\Delta$ is exponentially small with system size, regardless of $K$ in $n\le K\le M$ and distribution of $\bm{t}$. 
Also, as we only used the structure of the local parallel architecture, the result holds for any circuit ensemble. 
\end{proof}

Theorem \ref{theorem3} states that regardless of the number of input SMSV states and their configuration, most of the outcomes of GBS have zero probabilities which are easy to approximate under a certain degree of polynomial depth, for any circuit instances.
Moreover, for the case of the local parallel random circuit, we prove that the argument from the proof of Theorem \ref{theorem2} holds even for the GBS scheme, requiring additional depth to get out of a poorly anti-concentrated regime.

\begin{theorem}\label{theorem4}
(Local random ensemble) For $d$-dimensional local parallel random circuit of depth $D\le \alpha_1 n^{\frac{2(\gamma-1)}{d}-\lambda}$ with $\alpha_1 = c_{1}^{\frac{2}{d}}e^{\frac{2}{d}}\beta d/{2}^{\frac{2}{d}+3}$, for any $\lambda > 0$,$0<\beta<1$, $n\le K\le M$ and input mode configuration, it is easy to estimate the output probabilities of GBS within additive error $\epsilon = \rm{poly}\it{(n)^{-1}\frac{(2n)!}{M^{2n}}}$ for $1-\xi$ portion of output instances with probability $1-\delta$ over the random circuit instances, where $\xi$ and $\delta$ are exponentially small with system size.
\end{theorem}

\begin{proof}
See Appendix \ref{theorem4proof}.
\end{proof}

We remark that since the anti-concentration property is a necessary condition for the existing simulation hardness proof technique, the lack of anti-concentration itself does not guarantee the classical simulability of shallow-depth boson sampling in the local parallel architecture. 
Although there exist efficient classical algorithms that can simulate sparse output probability distribution of some circuit families~\cite{schwarz2013simulating, roga2020classical, bravyi2023classical}, they cannot be directly applied to our cases. 
Hence, our future work would be to demonstrate how the lack of anti-concentration in the local parallel architecture leads to the classical simulability of the shallow-depth boson sampling.

\section{Geometrically non-local architecture: hypercubic structure}\label{nonlocalhypercubicstructure}
So far, we have shown that the $d$-dimensional local parallel circuit is faced with the limit for achieving sampling hardness at shallow depth. 
As the size of the lightcone of the local architecture grows polynomially with depth, a mode connectivity issue arises below a certain polynomial depth, which leads to the prohibition of most of the outcomes and thus undermines the average-case hardness.
Also, as the random instances of the circuit are characterized by diffusive dynamics, most of the photon propagation is determined by the effective lightcone whose size is smaller than the original one, which requires additional depth to resolve the mode connectivity issue.
The problem is that from the photon loss model for each unit depth, polynomial depth entails large noise that disturbs classical intractability.

Since the problems we addressed come from circuit architecture which allows only local interactions, an obvious way to resolve those problems is to consider non-local interactions along modes, i.e., for the case that geometrically non-local unitary gates are available. 
Allowing the non-local interactions, we find a circuit that can mitigate the issues about connectivity and diffusive property within logarithmic circuit depth, where the architecture of the circuit has been used in numerical linear algebra in order to construct structured matrices efficiently \cite{parker1995random, mathieu2014fast, dao2019learning}.
Throughout this section, we refer to the circuit as a non-local hypercubic structure (NLHS) circuit for simplicity, and an example of the circuit architecture for mode number $M = 16$ is illustrated in Fig.~\ref{fig:partition circuit}.

\begin{figure*}[t]
    %\centering
\includegraphics[width=0.8\linewidth]{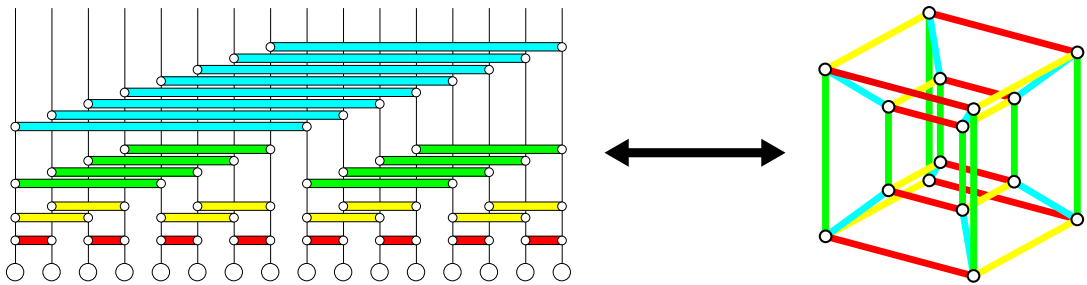}
\caption{A Schematic of a one-cycle of NLHS circuit for total mode number $M=2^4$. 
In this case, the architecture of the circuit can be interpreted as a 4$d$ hypercube, also known as a tesseract. 
%The architecture is composed of geometrically non-local unitary gates. 
Here, each color implies a depth unit, as all of the unitary gates with the same color are applied in parallel.}
\label{fig:partition circuit}
\end{figure*}

\subsection{Structure of the circuit}

Let the total mode number be $M = 2^p$ (for some integer $p$) for convenience. 
%\cor{We use a notation $d$ as the circuit can be constructed with geometrically local unitary gates if the circuit dimension is allowed up to $d$, which can easily be deduced from the hypercubic structure. }
The sequence for the construction of a one-cycle of NLHS circuit is simple:
for each $D = 1,2,\cdots,p$, apply unitary gates between mode number $2^D(j-1) + k$ and $2^D(j-1) + k + 2^{D-1}$, for all $j = 1,2,\cdots,2^{p-D}$ and $k = 1,2,\cdots,2^{D-1}$. 
Here, we point out that each $D$ implies the unit depth in terms of our standard, as unitary gates are applied in parallel for all $j$ and $k$ indices (e.g., all of the unitary gates with the same color in Fig.~\ref{fig:partition circuit} are applied in unit depth). 
%Indeed, the structure of the circuit architecture can be interpreted as a $p$-dimensional hypercube.
The corresponding circuit architecture has a well-spreading property such that photons departed from any mode can spread $2^D$ modes for depth $D$, and thus fully connected for the one-cycle of the circuit, viz., at least one path exists from any input mode toward any output mode. 
%Therefore, following the process, photons can spread full modes with $d = \log M$ depth, via the exponential growth of lightcone with depth by fully employing the geometrically non-local unitary gates. 
Additionally, even for arbitrary $M$, we can still achieve full mode connection in $\mathcal{O}(\log M)$ depth by iterating the above process (e.g., by stacking one-cycle of the circuit with depth $\lfloor \log M \rfloor$ on the left side and right side of the modes, alternatively).
Even so, to simplify the analysis, we only consider the case of $M = 2^p$ with a fixed $p\in\mathbb{N}$ for the rest of the section and define a single round of NLHS circuit as the one-cycle of the sequence we introduced, i.e., the procedure up to circuit depth $D = p$.

The NLHS circuit architecture indeed satisfies the minimum requirement for achieving average-case hardness in a low-depth regime, namely, having accessibility to all outcomes in a single round of the circuit.
Now we consider the case that the circuit is randomly drawn from a specific circuit ensemble, to examine whether the probability distribution can satisfy the condition of average-case hardness. 
Specifically, we use a typical setup such that all unitary gates composing the NLHS circuit as independently chosen random beam splitters, each drawn from Haar measure on U(2), similar to the local parallel random circuit case.  
From now on, we define the corresponding circuit as \textit{random} NLHS circuit.

We first find that a single round of the random NLHS circuit has evenly spreading property in terms of single-photon propagation, compared to the localized property (i.e., diffusive dynamics) of the local parallel random circuit. 
We can check this property with an absolute squared unitary element averaging over random circuit instances, which corresponds to the average of single-photon transition amplitude between modes. 
Let $U_{i,j}^{D}$ be a unitary element for the circuit with depth $D$.
Using properties of random beam splitters each drawn from Haar measure on U(2), an average of an absolute squared unitary element becomes
\begin{align}\label{onedesign}
\begin{split}
\E[|U_{i,j}^{p}|^2] &= \frac{1}{2}\E[|U_{i,j}^{p-1}|^2] + \frac{1}{2}\E[|U_{i,j\pm2^{p-1}}^{p-1}|^2] \\
&= \frac{1}{4}\E[|U_{i,j}^{p-2}|^2] + \frac{1}{4}\E[|U_{i,j\pm2^{p-2}}^{p-2}|^2] \\ &+ \frac{1}{4}\E[|U_{i,j\pm2^{p-1}}^{p-2}|^2] + \frac{1}{4}\E[|U_{i,j\pm2^{p-1}\pm2^{p-2}}^{p-2}|^2] \\
&= \cdots \\
&= \frac{1}{2^p}\sum_{j=1}^{M}\E[|U_{i,j}^{0}|^2] \\&= \frac{1}{2^p}
\end{split}
\end{align}
by averaging over the beam splitters for each depth, where $\pm$ signs behind index $j$ are determined by the binary notation of $j$. 
Interestingly, the result shows that a single round of random NLHS circuit can mimic the Haar unitary up to the first moment.

Also, if we repeat a single round of random NLHS circuit $\mathcal{C}$ times (i.e., total $\mathcal{C}\log M$ depth, iteratively stacking a single round of NLHS circuit of depth $D = \log M$), the number of paths between arbitrary modes scales as $M^{\mathcal{C}-1}$. 
To be more detailed, as a path between arbitrary modes is uniquely determined for a single round of NLHS circuit, the number of paths between arbitrary modes grows as $1, M, M^2, \cdots$ by repeating the circuit. 
%\co{Could you explain why in a single sentence. I know this is trivial but it'd be good to have an explanation.}.
%(i.e., stacking random NLHS circuits $\mathcal{C}$ times iteratively, with each random circuit drawn independently)
%In this situation, correlations between modes quickly flattened along the possible mode combinations because the degree of correlation depends on how much the paths toward those modes overlap.
%\co{I'm not sure what the previous sentence means..}
This exponential growth of possible paths motivates us to elucidate the properties of the random NLHS circuit and how the repetition of the corresponding circuit works on those properties. 
%Thus, hiding the output instances into random circuit instances would be possible by repeating the partition circuit. 
%path가 많으면 많을수록 correlation이 많은 path 적은 path가 모두 섞이기 때문 

For the following sections, we propose numerical evidence that the random NLHS circuit quickly converges to the global Haar random unitary (i.e., random unitary drawn from Haar measure on U($M$)) by repeating a single round of the random circuit drawn independently for each repetition. 
Specifically, we examine the probability distribution and entanglement generation of the random NLHS circuit with different repetition numbers and their convergence behavior toward the global Haar random unitary circuit.

\begin{figure*}[t]
    %\centering
    \includegraphics[width=0.85\linewidth]{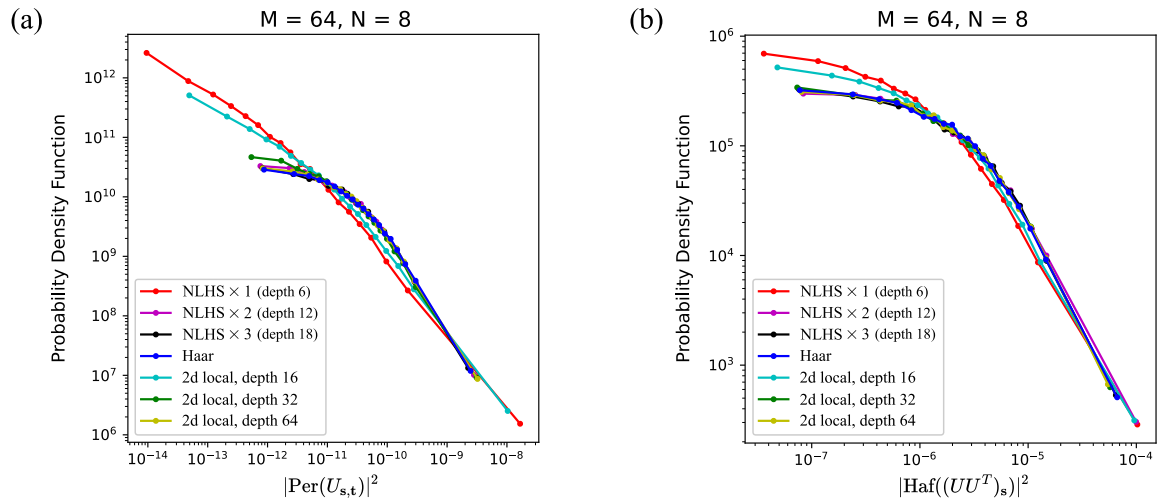}
    \caption{Probability density function for output probabilities of (a) FBS and (b) GBS with $M = 64$ and $N = 8$, for random NLHS circuit with different repetition numbers and 2-dimensional local parallel random circuit with different depths. The distribution corresponding to the global Haar unitary circuit is also displayed as an ideal case. The $x$-axis is for the (unnormalized) output probability values, and the y-axis is for the corresponding densities along the samples. Input $\bm{t}$ and output $\bm{s}$ are chosen randomly along possible configurations of $N$ photons along $M$ modes, without collision.}
    \label{fig:comparison}
\end{figure*}

\subsection{Probability distribution of random NLHS circuit}

%\co{This section is too lengthy. We should split this section properly into subsections. Like, the first part is probability distribution's property, and the last part is entropy something like that. Also, the explanation about the numerical study is too length. It would be better to make it compact. For now, it's hard to capture the main message.}

We investigate the output probability distribution of FBS and GBS over the random NLHS circuit (with repetitions) instances and its resemblance to the distribution from the global Haar random unitary. 
To statistically analyze the output probabilities over the randomly chosen circuit instances, we employ the \textit{probability density function}, which is a modified version of the histogram and employed by Ref.~\cite{aaronson2011computational} to numerically show evidence of the anti-concentration property of FBS (for more details, see Fig.~5 in Ref.~\cite{aaronson2011computational}).
The reason we use the probability density function is for enhanced visibility, as we observed that the degree of concentration of output probability distribution dramatically changes with circuit depth.

The difference from the histogram is, instead of using equal intervals on the density axis, each interval now contains equal numbers of samples. 
Specifically, samples sorted in ascending order are divided into each bucket, containing an equal number of samples.
Next, the density of each bucket is determined by the fraction of samples it contains divided by its width, where the width of the bucket is the difference between the maximum and minimum value of samples in the bucket. 
The probability density function is a plot of those densities corresponding to each bucket, where the $x$-axis represents the value of each sample. 
Throughout this section, we fix the number of output probability instances for each unitary circuit as 10000 and fix the number of buckets as 20, so a single point of the density function contains 500 output probability instances. 

%The problem is that a histogram of the probability for different circuit depths has limited visibility even for the logarithm scale, and the probability density function is useful for comparison between circuits with different depths.  

%As the probability over the random circuit instances is concentrated near the origin, a histogram of the probability has limited visibility, so the probability density function described above is useful for this situation. 

For output probabilities, we use $|\text{Per}(U_{\bm{s},\bm{t}})|^2$ for the FBS scheme, which corresponds to the output probability of FBS to get output $\bm{s}$ from input $\bm{t}$.
Also, we use $|\text{Haf}((UU^T)_{\bm{s}})|^2$ for the GBS scheme, which corresponds to the unnormalized output probability of GBS to get output $\bm{s}$ from full-mode SMSV input state with equal squeezing. 
To simplify the analysis, we fix the output photon number of GBS.

%{\vspace*{2ex} \noindent {\it Comparison with the local parallel circuit. }} 
\subsubsection{Comparison with the local parallel circuit}
Using the probability density function, we first compare the performance of the random NLHS circuit and typical 2-dimensional local parallel random circuit, for fixed mode number $M = 64$ and photon number $N = 8$. 
In this case, a single round of random NLHS circuit can be implemented with circuit depth $D = \log M = 6$, and thus $\mathcal{C}$ repetition of the circuit requires $6\mathcal{C}$ depth. 
Here, we consider the repetition number $\mathcal{C}$ up to three (i.e., up to $D = 18$), with the random circuit drawn \textit{independently} for each repetition. 
Also, for the 2-dimensional local parallel random circuit, we consider circuit depth from $D = 16$, not to restrict photon propagation between arbitrary modes.  
%such that there is no physical restriction of the photon propagation between arbitrary modes. 

We randomly sample 10000 unitary matrices for each circuit with different depths and calculate the output probability of randomly chosen collision-free input/output (only random collision-free output for GBS) for each unitary matrix we sampled. 
Here, we can focus on collision-free cases if $M$ is much larger than $N$ such that $N = \mathcal{O}(\sqrt{M})$ \cite{aaronson2011computational}, which fits well with our setup. 
Using those probability values, we plot the probability density function in Fig.~\ref{fig:comparison}, for random NLHS circuits with different repetition numbers, and for 2-dimensional local random circuits with different depths.
We also plot the function corresponding to $M$ dimensional Haar random unitary circuit, as an ideal case. 
The dispersion of the distribution along the $x$-axis implies a variance of output probability values over the random instances, such that the less dispersed the distribution implies the more the corresponding output probability distribution is anti-concentrated.

\begin{figure*}[t]
    %\centering
    \includegraphics[width=\linewidth]{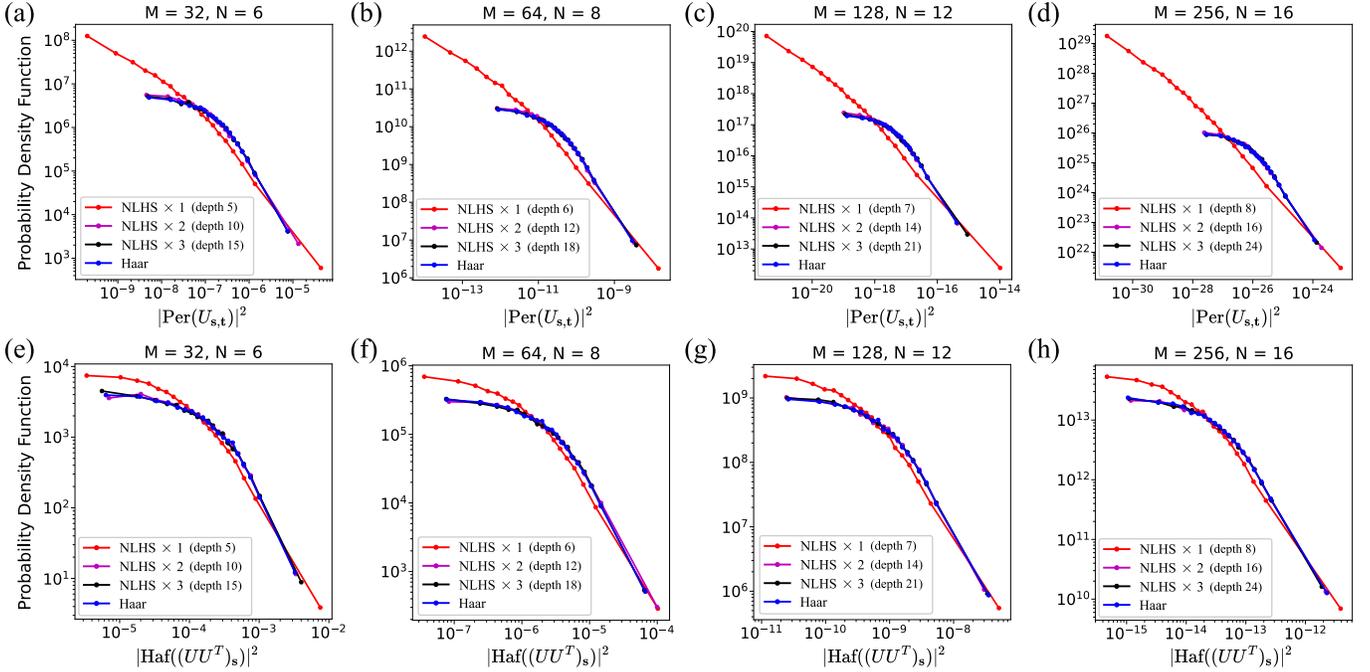}
    \caption{Probability density function for (a-d) FBS and (e-h) GBS, with $M = 32, 64, 128, 256$ and $N = 6, 8, 12, 16$ each. For each figure, the distributions corresponding to the random NLHS circuit, with repetition numbers up to three, are presented. Also, the distribution of the Haar unitary is displayed as an ideal case. Input $\bm{t}$ and output $\bm{s}$ are chosen randomly along possible configurations of $N$ photons along $M$ modes, without collision. 
    %The results show the number of repetitions required to imitate the distribution of Haar unitary does not scale with system size.
    }
    \label{fig:densitysum}
\end{figure*}

We find that distributions from both circuits converge to a distribution from $M$ dimensional Haar random unitary as increasing depth, which is predictable from the convergence properties studied at \cite{Emerson2003science, Emerson2005pra}. 
However, their convergence behavior toward global Haar random unitary is different. 
Specifically, although a single round of random NLHS circuit has poor performance in terms of anti-concentration, an iteration of the random NLHS circuit makes quick convergence to the distribution of the Haar random unitary, both for FBS and GBS schemes.

%{\vspace*{2ex} \noindent {\it Convergence to Haar measure for different system size. }}
\subsubsection{Convergence to the behavior of the global Haar measure, for different system sizes}
To identify how the convergence behavior varies as system size scales, we further investigate the number of repetitions of the random NLHS circuit required for the convergence to the global Haar random unitary with increasing system size. 
Here, we also examine the repetition number up to three with the random circuit drawn independently for each repetition. 
We set mode numbers $M = 32, 64, 128, 256$ and randomly sample 10000 unitary matrices for each repetition of the random NLHS circuit and for $M$ dimensional Haar random unitary circuit. 
Then we calculate the output probability of randomly chosen input/output for each matrix, both for FBS and GBS schemes, given output photon number $N$ around $\sqrt{M}$.
We plot the probability density function for those values in Fig.~\ref{fig:densitysum}.

Interestingly, the number of repetitions required to imitate the distribution from the global Haar random unitary is insensitive to system size. 
More specifically, stacking the random NLHS circuit twice dramatically changes its distribution, showing a close resemblance to the distribution from the global Haar unitary, which implies the existence of critical behavior between stacking the circuit once and twice.
Here, the mode numbers we used (up to 256 modes) cover all recent experimental results of GBS \cite{zhong2020quantum, zhong2021phase, Xanadu2022nature}, even though the photon numbers we used are small compared to those results.
Therefore, up to the system size covered by near-term experiments, the random NLHS circuit can imitate output probability distributions of fixed photon number $N \sim \sqrt{M}$ from the global Haar random unitary circuit with considerably low depth. 
%\co{I think it would be too much to say something about asymptotic limits because this is still finite-size numerical results.}
%Also, we expect that the number of repetitions required will not dramatically change as the system size scales.
%That is, the random NLHS circuit can imitate probability distributions from Haar unitary with considerably low depth, so it is directly applicable to near-term experiments. 
Also, as the required repetition number for imitation is insensitive to the system size, the result gives hope that the low-depth circuit (a few repetitions of the random NLHS circuit) can well approximate the Haar unitary and thus suggests evidence of the average-case hardness for larger quantum systems.

We also examine the scaling behavior of Haar measure convergence in terms of photon number $N$, for fixed mode number $M$. 
We plot the probability density function for fixed mode number $M = 32$ with output photon numbers $N = 8, 16$ both for FBS and GBS schemes, which can be checked in Appendix \ref{fixedmodenumber}.
The result shows similar behavior, i.e., the number of repetitions required does not change a lot as the output photon number increases. 
Hence, the result suggests a possibility such that a low-depth random NLHS circuit can even imitate output probability distributions from the global Haar random unitary for output photon number much larger than $\mathcal{O}(\sqrt{M})$, which is also remarkable.
%which is also applicable to near-term experiments where the output photon number is large. 

%{\vspace*{2ex} \noindent {\it Hiding property. }}
\subsubsection{Hiding property}
Additionally, we investigate if hiding property holds for the random NLHS circuit. 
The hiding theorem of BS shows two properties \cite{aaronson2011computational, deshpande2021quantum}.
One is that probability distribution is output independent, i.e., the probability of any outcome instances follows the same distribution over random circuit instances.
The other is that the corresponding output probability is similar to the quantity which is conjectured to be hard to estimate on average, viz., squared permanent (hafnian) of random Gaussian matrices.  
%\co{What do you mean by some?}.
We numerically investigate both of the above properties using the probability density function; readers who are interested in this subject can see the result in Appendix~\ref{hidingevidence}.

\begin{figure*}[t]
    \centering
    \includegraphics[width=\linewidth]{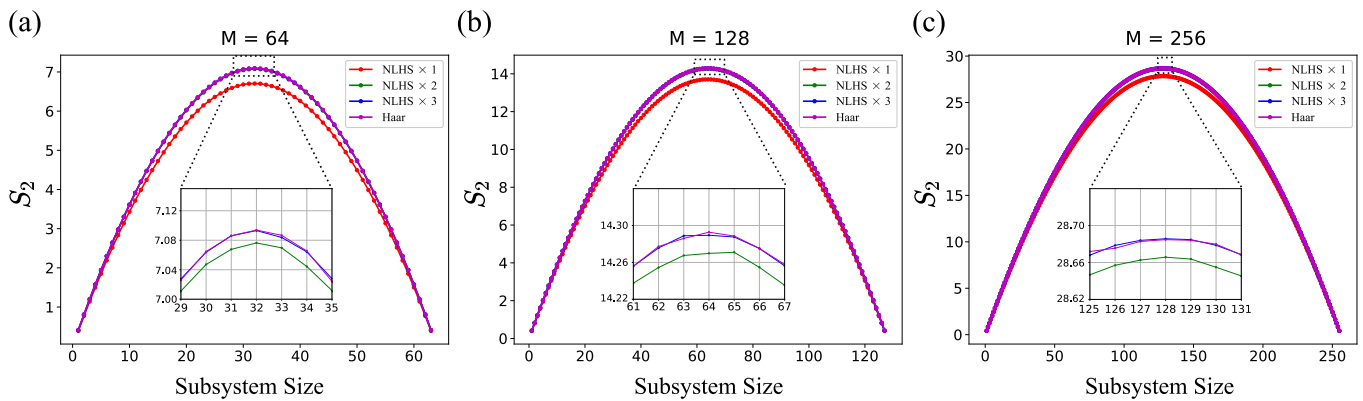}
    \caption{The average of R\'enyi-2 entropy with respect to different subsystem sizes, for mode number (a) $M = 64$, (b) $M = 128$, (c) $M = 256$.  For each figure, the average of entropy values corresponding to the Gaussian state evolved by various repetitions of random NLHS circuits and $M$ mode Haar random unitary circuits, are presented. 
    %The result shows quick convergence to Haar random unitary and insensitivity of the converging speed with increasing system size. Moreover, the absolute difference of the entropy values from Haar unitary varies very slowly with system size.
    }
    \label{fig:page}
\end{figure*}

\subsection{Entanglement generation of random NLHS circuit}
In this section, we investigate how entanglement along the modes varies with the repetition of the NLHS circuit, to show the simulation hardness and the convergence behavior toward the global Haar measure from another perspective. 
Entanglement is considered necessary for the hardness of classical simulation of a quantum state, as various tensor network methods can approximate the low entangled state efficiently \cite{Vidal2003prl, Vidal2004prl, Verstraete2004prl}.
R\'enyi entropies are often cited as a measure of entanglement, and possibly indicate the feasibility of simulating the given quantum state via those methods~\cite{Verstraete2006prb, Schuch2008prl}.
Specifically, restricting to bosonic Gaussian states, many results \cite{Adesso2012prl, Lami2016prl, Camilo2019prb} suggested that R\'enyi-2 entropy is a good measure of entanglement. 
Moreover, Ref.~\cite{Iosue2022page} recently proposed the R\'enyi-2 Page curve of pure Gaussian state evolved by the global Haar random unitary, exactly the output state of the ideal GBS scheme. 
%Also, Ref.~\cite{Iosue2022page} verified that R\'enyi-2 entropy $S_2$ of the reduced state gives a bound (including both upper and lower bound) for the entanglement entropy of random pure Gaussian state, specifically the output state of the GBS scheme. 
Hence, to examine the entanglement behavior, we focus on the R\'enyi-2 entropy of reduced states of the output states of GBS, each evolved by different repetitions of random NLHS circuits or global Haar random unitary.
%where the quantum state corresponds to the $M$ product of SMSV with equal squeezing evolved by random NLHS circuits or Haar random unitary.
Our goal is to examine the entanglement generation of the random NLHS circuit with increasing depth and investigate if the circuit with increasing depth can reproduce the R\'enyi-2 Page curve, to propose evidence of the simulation hardness of GBS for the corresponding circuit and its convergence behavior toward the global Haar random unitary. 
%We examine the convergence to global Haar measure and simulation hardness of the random NLHS circuit, by comparing R\'enyi-2 entropy from random NLHS circuits with R\'enyi-2 Page curve from Haar random unitary.
%Our goal is to reproduce the Page curve~\cite{Page1993prl} for the unitary of the random NLHS circuit with different stacking numbers and Haar random unitary, to figure out the existence of the convergence behavior and the simulation hardness of GBS for the corresponding circuit.

For mode numbers $M = 64, 128, 256$, we randomly sample 10000 unitary matrices from each repetition of the random NLHS circuit, and also from the global Haar measure as an ideal case.
Then the $M$ product of SMSV with equal squeezing parameter $r = 0.4$ is evolved with those unitary matrices.
To calculate the entropy of the output states, $M$ modes are partitioned into two groups, one for $k$ modes and the other for $M-k$ modes, where the mode selection is completely random along ${}_{M}C_{k}$ possible combinations. 
For all $k\in [M-1]$, we calculate the entropy of a reduced state and average over the unitary matrices, while the R\'enyi-2 entropy takes the form of $S_2 = \frac{1}{2}\log \det \sigma$, where $\sigma$ is the covariance matrix of the reduced state \cite{serafini2017quantumcontinuousvariable}.

In Fig.~\ref{fig:page}, we plot the average of R\'enyi-2 entropy with respect to different subsystem sizes, for each mode number we fixed. 
The result shows the entanglement generation of the random NLHS circuit with an increasing stacking number, where the distribution converges to the R\'enyi-2 Page curve (i.e., distribution from the global Haar random unitary) as circuit depth increases. 
It is notable that the required number of repetitions for the convergence is insensitive to system size, similar to previous results we addressed. 
Specifically, stacking the random NLHS circuit twice or more resembles the behavior of the global Haar measure, and the absolute difference of the entropy values between them varies very slowly with system size, which is also remarkable.

\subsection{Unitary design of random NLHS circuit}
We also investigate the convergence behavior of the random NLHS circuit toward global Haar random circuits from a different perspective, which is unitary design. 
As we can see in Eq.~\eqref{onedesign}, even a single round of random NLHS circuit can mimic the Haar random unitary up to the first moment. 
Based on this understanding, we aim to explore if the random NLHS circuit can still mimic the Haar random unitary for the higher moments, i.e., if it has a unitary design property. 
Among the various measures of unitary design as in~\cite{low2010pseudo, sim2019expressibility}, we focus on the frame potential~\cite{gross2007evenly} of the random NLHS circuit, which is suitable for numerical evaluation and often cited as evidence of the unitary design property~\cite{roberts2017chaos, cotler2017chaos, liu2018spectral, liu2022estimating}. 
%To do so, we numerically evaluate the frame potential~\cite{gross2007evenly} for the random NLHS circuit, which can measure the closeness to the Haar measure and is often cited as evidence of a unitary design property~\cite{roberts2017chaos, cotler2017chaos, liu2018spectral, liu2022estimating}. 
Here, we also calculate the frame potential for the local parallel random circuits with various circuit dimensions to compare their convergence behaviors with the random NLHS circuit.
We find that the random NLHS circuit shows fast convergence behavior with increasing circuit depth for higher moments $k = 2, 3, 4$, and its convergence behavior is faster than those of other local circuits.
This suggests evidence of the unitary design property of the random NLHS circuit within a few repetitions. 
The details are provided in Appendix~\ref{framepotential}.

\subsection{Experimental realization of NLHS circuit}
Throughout the previous sections, we have numerically shown that the random NLHS circuit shows quick convergence behavior toward Haar unitary distribution by a few repetitions, which gives a possibility for the average-case hardness of BS with shallow-depth quantum circuits.
In this section, we discuss the experimental feasibility of the random NLHS circuit. 
Experimental realization of the random NLHS circuit (with repetitions) requires a geometrically non-local setup or a high-dimensional architecture with its dimension up to $d = \log M$. 
%However, for the experimental realization of the random NLHS circuit, one problem we face is that our circuit architecture requires free usage of geometrically non-local gates. 
%This setup is considered hard to implement with photonic systems. 
%\co{Unpromsing is too pessimistic. Also, there might have been some attempts we missed.}

Considering photonic systems, Ref.~\cite{crespi2016suppression} implemented integrated photonic chips for $2^d$ modes with dimensions up to $d = 3$, where the interaction along modes occurs in hypercubic sequence, precisely the circuit architecture of the NLHS circuit. 
Hence, implementation of the random NLHS circuit would be feasible if the dimension of those devices can be increased in a scalable manner and if the gates composing the devices are programmable. 
Additionally, we can employ high-dimensional photonic architecture used in \cite{deshpande2021quantum, Xanadu2022nature} but in a different gate sequence, i.e., the hypercubic sequence we used. 
Optical frequency crystals, as demonstrated in \cite{imany2020probing, hu2020realization}, would be a viable alternative, which can also suggest high-dimensional photonic architecture as the NLHS circuit. 
%Although the scalability is unclear, experimental realization of large size NLHS circuit is possible if we can increase the size of those devices
%Also, Refs.~\cite{Louri1994hypercube, Pinkston1994} presented schematics of optical hypercubic networks, and Ref.~\cite{de2022measurement} implemented hypercubic networks up to which can be utilized as an architecture for a single round of NLHS circuit. 
%\cite{de2022measurement} using temporal modes.

%The implementation of the random NLHS circuit can also be achieved through phononic systems. 
We can also consider phononic systems, such as trapped ions, which can be utilized as alternative bosons to construct linear-optical circuits \cite{lau2012proposal, shen2014scalable, leibfried2003quantum}. 
Specifically, recent works from \cite{Chen2023scalable} have experimentally demonstrated a programmable all-to-all setup for four modes of phononic network with trapped ions, using collective-vibrational modes of ions. 
Although architectures based on trapped ions might face the limitation of long-range interaction due to spectral crowding when system size scales, employing modular architecture would help overcome this problem. 
%This approach suggests that a programmable beam splitter between arbitrary pairs of modes can be realized in a scalable manner; such properties are necessary for the implementation of the random NLHS circuit with increasing system size. 
%Although there exists a limitation to the number of ions due to the spectral crowding 
Therefore, this setup may also offer a promising candidate for the experimental realization of the random NLHS circuit in the near future.

\section{Discussion}\label{discussion}
In this work, we examined that for local circuit architecture, the depth of the linear optical circuit should be large enough to satisfy the conditions for the sampling hardness based on the currently available proof technique. 
More precisely, for depth under the degree of $\mathcal{O}(M^{\frac{(\gamma - 1)}{\gamma d}})$, most of the probabilities are zero regardless of circuit ensemble and input configuration. 
Besides, for local random ensemble, most of the probabilities are too small and easy to estimate for depth under $\mathcal{O}(M^{\frac{1}{\gamma}[\frac{2(\gamma - 1)}{d} - \lambda ]})$.
We interpreted that this problem comes from the properties of the local parallel architecture, where the circuit lightcone grows polynomially with depth, and the circuit is characterized by the diffusive property, which hinders anti-concentration. 
We proposed the circuit architecture using non-local interactions, which can restrain the issues we addressed at logarithm circuit depth. 
Also, the repetition of the random NLHS circuit shows quick convergence behavior toward the global Haar random unitary, where the speed of convergence is insensitive to the system size.
Hence, we conclude that the corresponding circuit can be used for approximate Haar measure with shallow depth circuit, and has the potential to be utilized as an architecture for scalable quantum advantage with BS. 

Here are a few remarks about our results and related open questions:

1. As we have described in Sec.~\ref{averagecasehardnessintroduction}, anti-concentration is a necessary condition for the ``current" average-case hardness proof technique.
This means that the lack of anti-concentration in Sec.~\ref{limitationsoflocal} does not directly imply the classical simulability of boson sampling in shallow-depth local parallel architecture.
Accordingly, it does not completely rule out the classical intractability of boson sampling in local shallow-depth architecture, although a new proof technique for the simulation hardness has to be developed in this case.
%Therefore, the remaining challenge to be addressed is to demonstrate how the lack of anti-concentration leads to the classical simulability of the shallow-depth boson sampling. 
To fully investigate whether the classical simulation in local shallow-depth circuits is easy or hard, we can take two different approaches: $(i)$ demonstrate how the lack of anti-concentration leads to the classical simulability of the shallow-depth boson sampling, or $(ii)$ find average-case hardness \textit{inside} the exponentially small portion of outcomes. 
We leave those problems as open questions.

2. It is worth emphasizing the difference of our result from the one in Ref.~\cite{deshpande2021quantum}, which provides evidence of the hardness of sampling from the high dimensional circuit of {\it constant} depth. 
We stress that our depth unit is different from the depth definition in Refs.~\cite{deshpande2021quantum, Xanadu2022nature}, where the difference comes from the structural difference of the circuit architecture.
Specifically, they denote unit depth by one cycle of serial application of gates along the modes for each dimension, which translates to depth $\mathcal{O}(M)$ from our standard such that only parallel implementation of gates is allowed for unit depth.
%\co{we should emphasize how this definition translates to our depth unit, like depth $M$ in our definition?}, compared to our standard 

3. The key difference between RCS and BS with shallow-depth quantum circuits lies in the accessibility of outcomes, which comes from a systematic difference.
%between the optical (continuous-variable) and the qubit (discrete-variable) systems. 
In RCS, as each outcome is obtained through measurement on a computational basis, the local Pauli-X operator on each qubit gives access to all possible outcomes. 
Hence, when a circuit for RCS is constructed using local Haar unitaries, all of the outcomes are accessible regardless of the circuit depth, and output symmetry from random circuit instances is easily established by the translation invariance of the Haar measure over the local random unitary gates. 
%and the problem of whether the classical simulation is possible depends on its probability distribution along the outcomes. 
However, in the case of BS, each outcome is obtained through measurement on a photon number basis; the accessibility of outcomes can be restricted in low-depth regimes as we have shown. 
Since the current proof of the sampling hardness requires the usage of all possible outcomes (as the allowed error is given by total variation distance), we emphasize that the circuit architecture for BS must be carefully designed.
In this regard, the NLHS circuit we proposed may play an important role in the shallow-depth hardness argument as it gives access to all possible outcomes at logarithm depth.

4. As we discussed in Sec. \ref{averagecasehardnessintroduction}, the measure over $M$ mode random NLHS circuits with a few repetitions cannot strictly imitate the Haar measure on U($M$) unless the repetition number satisfies $\Omega(M)$. 
%This can be easily checked by free parameter counting.
%The number of free parameters for a $M$ dimensional unitary matrix is $\mathcal{O}(M^2)$, and specifically this reduces to $\mathcal{O}(M^{\frac{\gamma + 1}{\gamma}})$ for the FBS case with fixed input configuration. In contrast, the number of free parameters for $M$-mode NLHS circuit is $\mathcal{O}(M\log M)$ as there are $\mathcal{O}(M)$ free parameters for each unit depth.  
%\cor{The lack of free parameters implies that a distribution of the circuit matrix does not contain all possible unitary matrices, viz., it is only a subset of the unitary group.}   
%herefore, to exactly implement the $M$-mode Haar unitary experimentally, at least a polynomial repetition of the circuit is required. 
However, even though there is only numerical evidence yet, the repetition of the random NLHS circuit can well approximate certain properties (output probability distribution, entanglement generation) of the global Haar random unitary.
This implies that \textit{approximating} Haar measure might be possible only with a few repetitions of the circuit \cite{Emerson2003science, Emerson2005pra}. 
Our results raise an open question of how much the measure over the random circuit matrices can imitate the global Haar measure, with increasing repetition numbers.
%a similar argument from the continuous-variable version of the unitary design.  }

5. A few repetitions of the random NLHS circuit show global randomness for fixed-size numerical experiments, which gives a possibility to show the average-case hardness of output probability approximation at a shallow depth regime, for the asymptotic limit. 
As we discussed in the introduction, there is a possibility that a constant number of repetitions of the random NLHS circuit (i.e., logarithm depth circuit) may enable us to avoid the classically simulable regime of noisy BS with photon loss.
Specifically, as the output photon number is given by $N_{\text{out}}=NT^{D}$ for circuit depth $D$, asymptotically $N_{\text{out}}$ can be larger than $\mathcal{O}(\sqrt{N})$ for $D = \mathcal{O}(\log N)$. 
%For an example of a photon loss noise model,  for the case that constant transmission rate $T$ is applied at each depth. 
%Specifically, if a constant number of repetitions is enough 
%For the case that noises are applied at each depth, such that noise rate is not excessively enlarged 
%Specifically, if a constant or maximally logarithmic number of repetitions of the random NLHS circuit is enough to show global randomness asymptotically,  
Therefore, another open problem would be to show the average-case hardness of approximating output probabilities from a few repetitions of the random NLHS circuit, for the case that noises are applied at each depth.

\begin{acknowledgements}
%We thank all for the interesting and fruitful discussions.
B.G. and H.J. were supported by the National Research Foundation of Korea (NRF) grants funded by the Korean government (Grant~Nos.~NRF-2020R1A2C1008609, 2023R1A2C1006115 and NRF-2022M3E4A1076099) via the Institute of Applied Physics at Seoul National University, and by the Institute of Information \& Communications Technology Planning $\&$ Evaluation (IITP) grant funded by the Korea government (MSIT) (IITP-2021-0-01059 and IITP-2023-2020-0-01606).
C.O. and L.J. acknowledge support from the ARO (W911NF-23-1-0077), ARO MURI (W911NF-21-1-0325), AFOSR MURI (FA9550-19-1-0399, FA9550-21-1-0209), AFRL (FA8649-21-P-0781), NSF (OMA-1936118, ERC-1941583, OMA-2137642), NTT Research, and the Packard Foundation (2020-71479).

\end{acknowledgements}

\begin{appendix}
\section{Proof of Theorem \ref{theorem2}}\label{theorem2proof}

\begin{figure*}[t!]
    %\centering
    \includegraphics[width=\linewidth]{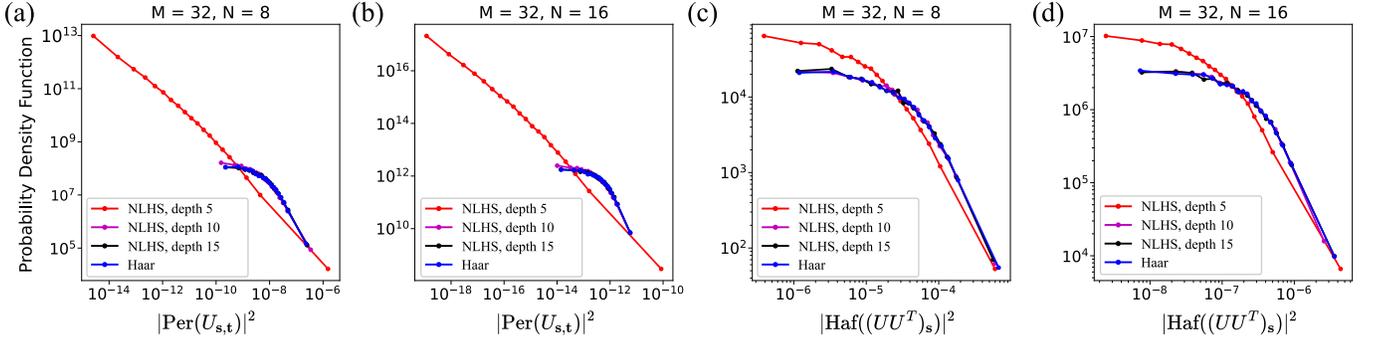}
    \caption{Probability density function for (a-b) FBS and (c-d) GBS, with $M = 32$ and $N = 8, 16$ each, Input $\bm{t}$ and output $\bm{s}$ are fixed at first $N$ modes over $M$ modes. The distributions corresponding to the random NLHS circuit, with repetition numbers up to three, are presented. Also, the distribution of the Haar unitary is displayed as an ideal case. %The results show that the number of repetitions required to imitate the distribution of Haar unitary does not scale with increasing photon number.
    }
    \label{fig:densitysumfixedmode}
\end{figure*}

In this Appendix, we provide the proof of Theorem~\ref{theorem2}.
\begin{proof}
We first follow the process of approximating the output probability distribution of the local parallel random circuit using unitary matrix truncation; more details can be found in Refs.~\cite{Zhang2021npjq, Oh2022PRL}. 
It is known that photons follow classical random walk behavior on average in the local parallel random circuit. 
From this behavior, they are effectively localized in a regime even smaller than a given lightcone; we can find a bound of the regime such that the probability of leakage from this regime is exponentially small. 
%This means the farther the photon propagates from the source, the less likely it will happen. 
In this case, truncation of the circuit unitary matrix by discarding matrix elements outside the regime results in an exponentially small total variation distance from the original output distribution.
This implies that the summation of output probabilities that at least one photon propagates outside the regime would also be exponentially small.

%To quantify this property, 
To be more specific, we define a leakage rate from the input mode $i$ as $\eta_{i}(l) \coloneqq \sum_j^{'}|U_{j,i}|^2$, where the summation is over the all possible modes that are geometrically away from the mode $i$ more than length $l$ for each dimension. 
From the classical random walk behavior, when the length $l$ satisfies the corresponding bound
\begin{equation}\label{leakagelength}
l \ge \sqrt{\frac{2N^{\lambda}D}{\beta d}}, 
\end{equation}
with any constant $\lambda > 0$ and $0 < \beta < 1$, then $\Pr\left[ \eta_i(l) \le \exp(-N^{\lambda})\right] \ge 1-\delta$ where $\delta \le 2d\exp\left[(1-\frac{1}{\beta})N^{\lambda}\right]$, and the probability is over the random instances of the circuit.

We define a matrix $\Tilde{U}$, a truncated version of a unitary matrix $U$, by discarding the matrix elements that are farther than $l$ for each dimension from the given column index. 
From this definition, $\Tilde{U} \equiv U - dU$ with $\|dU\|_{F}^{2} = \sum_{i}\eta_{i}(l)$, where the summation is over all possible column indices. If $l$ satisfies Eq.~(\ref{leakagelength}), then $\|dU\|_{F}^{2} \le \text{poly}(N)\exp(-N^{\lambda})$ with probabilty $1-\delta$ over the circuit instances. 
Using this fact with results from \cite{Arkhipov2015PRA}, we can deduce that the total variation distance between distributions from $U$ and $\Tilde{U}$ is bounded by $\poly(N)\|dU\|_{F}^{2}$ which is exponentially small with system size for our case. 
Hence, the summation of probabilities of outcomes corresponding to the discarded elements, i.e. the outcomes that at least a single photon propagates from the source more than $l$ satisfying Eq.~(\ref{leakagelength}), is bounded by the total variation distance which is exponentially small with high probability.

However, as $\Tilde{U}$ is not a unitary matrix, the output distribution of the matrix cannot be determined. 
We can resolve this issue by first extending $\Tilde{U}$ to a unitary matrix in $2M \times 2M$, which contains normalized $\Tilde{U}$ at first $M$ modes, and post-selecting the outcomes only at the first $M$ modes. 
This way, we can get the probability distribution corresponding to $\Tilde{U}$ (see \cite{Oh2022PRL} for more details).
To summarize, the summation of probabilities of the outcomes outside the regime $l$ has the order of maximally $\text{poly}(N)\exp(-N^{\lambda})$ with probability $1-\delta$ over the circuit instances.

We define an \textit{effective} lightcone $L_{D}^{'}$ such that the size of the lightcone for each dimension has a minimum value satisfying Eq.~(\ref{leakagelength}). 
Following the above discussions, the summation of probabilities of the outcomes outside the effective lightcone is exponentially small with high probability. 
Hence, if the portion of outcomes inside the effective lightcone is exponentially small, the output probability distribution is concentrated on the small portion of the outcomes. 
%even for the depth $D \ge \kappa_0N^{\frac{\gamma - 1}{d}}$ we find that the output probability distribution is still concentrated for the local parallel random circuit, i.e., each photon propagates only inside the effective lightcone with high probability.  

Let $\mathcal{M}'$ number of outcomes $\bm{s}$ satisfying the constraint (\ref{constraint1}), but now the lightcone term is replaced by the effective lightcone $L_{D}^{'}(t_i)$. Upper bound of $|L_{D}^{'}(t_i)|$ is
\begin{equation}
 \left|L_D^{'}(t_i)\right| \le \left(\frac{2N^{\lambda}D}{\beta d}\right)^{\frac{d}{2}},
\end{equation}
where equality holds when the effective lightcone does not meet geometrical boundaries. The ratio of outcomes inside the effective lightcone over all possible outcomes $\Delta'$ is
\begin{align}
    \Delta' &= \frac{\mathcal{M}'}{\binom{M+N-1}{N}} \\
&\le \sqrt{2\pi N}e^{\frac{1}{12N}}\left[\frac{N}{eM}\left(\frac{2N^{\lambda}D}{\beta d}\right)^{\frac{d}{2}}\right]^N \\
&< 3\sqrt{N}\left[ \frac{1}{ec_0}\left( \frac{2}{\beta d}\right)^{\frac{d}{2}}D^{\frac{d}{2}}N^{1-\gamma+\frac{\lambda d}{2}}\right]^N.
\end{align}
Indeed, for the depth $D \le \alpha_0N^{\frac{2(\gamma -1)}{d}-\lambda}$ with a constant $\alpha_0 = e^{\frac{2}{d}}c_{0}^{\frac{2}{d}}\beta d/{2}$, most of the probability distribution is concentrated inside the exponentially small portions of outcomes.

\begin{figure*}[bt!]
    %\centering
    \includegraphics[width=0.95\linewidth]{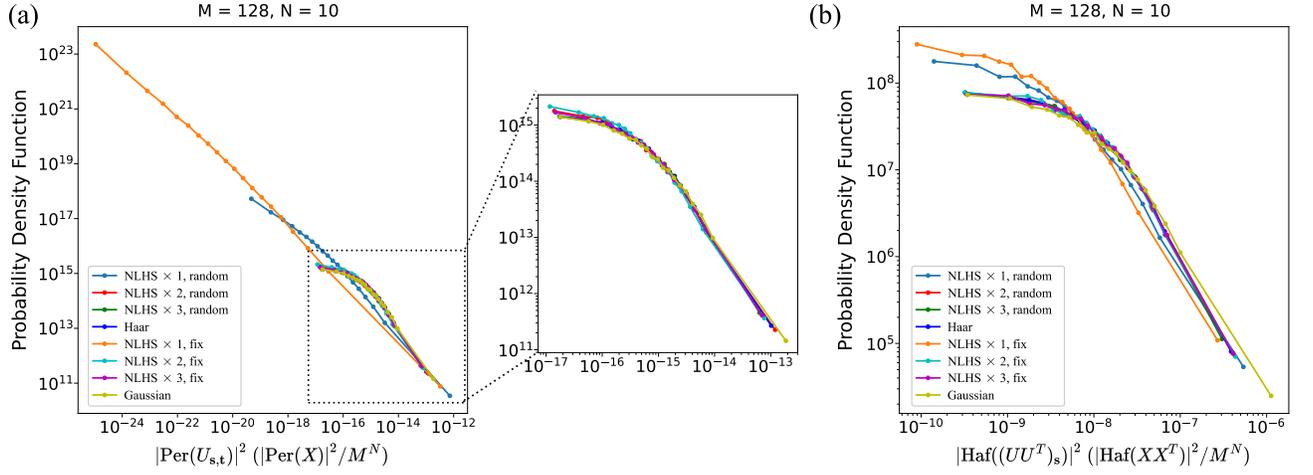}
    \caption{Probability density function for (a) FBS and (b) GBS with $M = 128$ and $N = 10$ each, for fixed (first $N$ modes over $M$ modes) and random input $\bm{t}$ and output $\bm{s}$ without collision. The distributions corresponding to the random NLHS circuit, with repetition numbers up to three, are presented. Also, the distributions of the Haar random matrices and random Gaussian matrices are displayed, to check if there exists the hiding property. 
    %The result presents the existence of the hiding property on some level after stacking the random NLHS circuit twice, both for FBS and GBS schemes.
    }
    \label{fig:hiding}
\end{figure*}

Our proof of Theorem \ref{theorem2} is not completed yet because the statement that the summation of truncated probabilities is smaller than $\text{poly}(N)\exp(-N^{\lambda})$ does not guarantee that all single truncated probabilities are small enough to estimate easily. 
There may exist some truncated probabilities larger than the allowed additive error $\epsilon$, which we cannot safely estimate as zero.

Nevertheless, we can verify that the truncated outcomes outside the effective lightcones having probabilities larger than given $\epsilon$ cannot occupy large portions of outcomes. 
Let $\Gamma$ be the ratio of truncated outcomes having probabilities larger than $\epsilon$ over all possible outcomes of FBS.
Also, let $\mathcal{S}'$ set of outcomes outside the effective lightcone, so $|\mathcal{S}'| = (1-\Delta')\binom{M+N-1}{N}$. 

We find the upper bound of $\Gamma$ as
\begin{align}
    \Gamma &\le \Pr_{\bm{s} \in \mathcal{S}'}\left[p_{\bm{s}} \ge \epsilon \right] \\
&\le \frac{1}{\epsilon}\E_{\bm{s} \in \mathcal{S}'}\left[p_{\bm{s}}\right] \label{eq15} \\
&= \frac{1}{\epsilon}\frac{\sum_{\bm{s} \in \mathcal{S}'}p_{\bm{s}}}{(1-\Delta')\binom{M+N-1}{N}} \\
&\le \frac{1}{\epsilon}\frac{N!}{M^N}\poly(N)\exp(-N^{\lambda}), \label{eq17}
\end{align}
where we used Markov's inequality for Eq.~\eqref{eq15} and the inequality Eq.~\eqref{eq17} holds with probability $1-\delta$ over the circuit instance. Therefore, for $\epsilon = \poly(N)^{-1}\frac{N!}{M^N}$, $\Gamma \le \poly(N)\exp(-N^{\lambda})$ with probability $1-\delta$ over the circuit instances.
%Nevertheless, we verify that the truncated outcomes with large probabilities cannot occupy large portions of outcomes, from the contradiction to the proposition that the summation of those probabilities is exponentially small. 
%We now use a notation for the corresponding ratio as $\Gamma$. 
%Let's return to the hardness argument of approximate sampling for quantitative analysis. 
%total variation distance for approximate simulation hardness requires its size of maximal $\poly(N^{-1})$, so the additive error allowed for each outcome has the order of $\poly(N^{-1})|\mathcal{S}|^{-1}$. 
%Thus, if a probability has the size of $\poly(N^{-1})|\mathcal{S}|^{-1}$, then it is not certain that we can approximate this to zero within additive error. 
%But for the case that the size of $\Gamma$ has the order of $\poly(N^{-1})$ or larger, i.e., truncated outcomes with large probabilities (order of $\poly(N^{-1})|\mathcal{S}|^{-1}$ or larger) occupy at least $\poly(N^{-1})$ portion of all possible outcomes, summation of those probabilities has an order of $\poly(N^{-1})$ which is a contradiction to the statement that it has an order of $\text{poly}(N)\exp(-N^{\lambda})$. 
%Therefore, $\Gamma$ should be exponentially small with system size, up to the order of maximally $\text{poly}(N)\exp(-N^{\lambda})$.

In summary,
%from Theorem \ref{theorem1}, the portion of permitted outcomes $\Delta$ is exponentially small for circuit depth $D \le \kappa_0N^{\frac{\gamma - 1}{d}}$ with a constant $\kappa_0 = c_{0}^{\frac{1}{d}}e^{\frac{1}{d}}d/{2}$, so we can exactly estimate most of the probabilities with zero. Also, 
for depth $D \le \alpha_0N^{\frac{2(\gamma -1)}{d}-\lambda}$ with a constant $\alpha_0 = e^{\frac{2}{d}}c_{0}^{\frac{2}{d}}\beta d/{2}$, the ratio of outcomes we cannot easily estimate their probabilities within additive error $\epsilon$, which can be characterized by $\xi \equiv \Delta' + \Gamma$, is exponentially small with probability larger than $1-\delta$ over the circuit instances. This completes the proof. 
\end{proof}

\begin{figure*}[bt!]
    %\centering
    \includegraphics[width=0.95\linewidth]{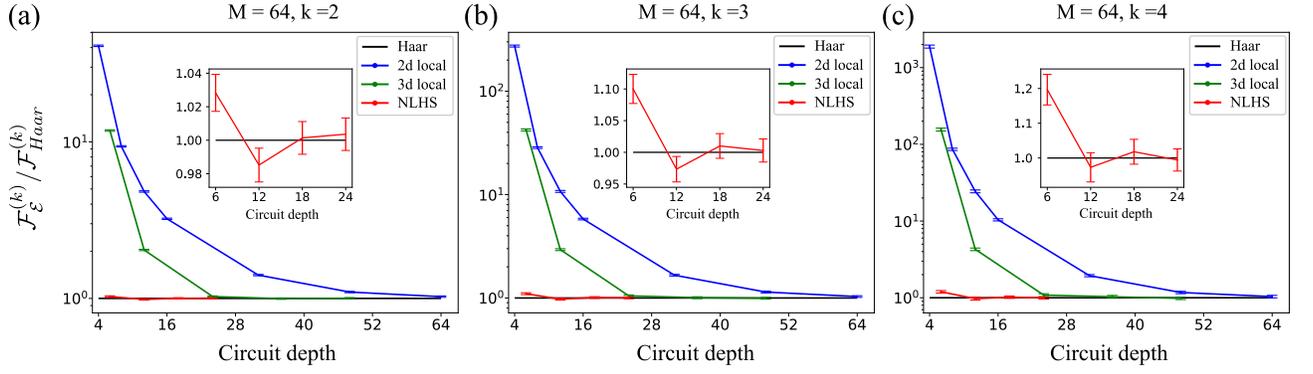}
    \caption{Normalized frame potential $\mathcal{F}_{\mathcal{E}}^{(k)}/\mathcal{F}_{\text{Haar}}^{(k)}$ for the moment (a) $k = 2$, (b) $k = 3$ and (c) $k = 4$, with $M = 64$.
    The distributions $\mathcal{E}$ corresponding to the random NLHS circuit, and the local parallel random circuits are presented. 
    The error bars at the plot correspond to the standard deviations induced by the bootstrapping method~\cite{efron1994introduction}. 
    Also, the theoretical value for the Haar random unitary is presented as a black horizontal line. 
    }
    \label{fig:framepotential}
\end{figure*}

\section{Proof of Theorem \ref{theorem4}}\label{theorem4proof}
In this Appendix, we provide the proof of Theorem \ref{theorem4}.
\begin{proof}
%We only consider the depth $D$ over $\kappa_1n^{\frac{\gamma - 1}{d}}$ because it is trivial below the corresponding depth from the result of Theorem \ref{theorem3}. 
We use the concept of effective lightcone as well for GBS, except for the notation change from $N$ to $n$. Similar to the FBS case, using truncation of unitary matrix by the effective lightcone, summation of probabilities of the outcomes outside the effective lightcone has the order of maximally $\text{poly}(n)\exp(-n^{\lambda})$ with probability $1-\delta$ over the circuit instances. Here, $\delta \le 2d\exp\left[(1-\frac{1}{\beta})n^{\lambda}\right]$ with any $\lambda > 0$ and $0<\beta<1$.

Here, the truncated outcomes can have any output photon numbers, contrary to our scheme that focuses on fixed output photon number $2n$. 
Hence, the summation of truncated probabilities can be enlarged from $\text{poly}(n)\exp(-n^{\lambda})$ if we normalize the output probability distribution in the $2n$ photon subspace.
However, it is not problematic if we also set the mean photon number $2n$, as the output probability of getting the mean photon number is at least an inverse polynomial of system size.
Specifically, the probability of generating $2n$ output photon events is given by negative binomial distribution \cite{hamilton2017gaussian} 
\begin{align}
P(2n) = \binom{K/2 + n - 1}{n}\frac{\tanh^{2n}r}{\cosh^{K}r}. 
\end{align}
By applying $2n = K\sinh^2{r}$, the above probability reduces to $P(2n) = \Omega(n^{-\frac{1}{2}})$, which is indeed an inverse polynomial.
%\begin{align}
%&P(2n) = \binom{M/2 + n - 1}{n}\frac{\tanh^{2n}r}{\cosh^{M}r} \\
%&> \frac{1}{3\sqrt{n}}\exp\left[\sum_{k=1}^{\infty}\frac{(-1)^{k-1}2^k}{kM^k}\left(\sum_{j=1}^{n-1}j^k - \frac{1}{k+1}n^{k+1}\right) \right]
%\end{align}
Hence, even considering only the $2n$ photon outcomes, the summation of probabilities truncated by effective lightcone still have the order of maximally $\text{poly}(n)\exp(-n^{\lambda})$.

From the above discussions, if the ratio of outcomes inside the effective lightcone is exponentially small, this means that the output probability distribution is concentrated on the exponentially small portion of the outcomes.
Let $\mathcal{M}'$ number of outcomes satisfying (\ref{constraint2}), but now the lightcone is replaced by the effective lightcone $L_{D}^{'}$. We can similarly use Eq.~\eqref{numberofpossibleoutcome2} but now Eq.~\eqref{upperboundoflightcone2} is modified to
\begin{equation}
 \left|L_D^{'} \circ {L_D^{'}}^{\dag}(r_i)\right| \le 2^d\left(\frac{2n^{\lambda}D}{\beta d}\right)^{\frac{d}{2}}.
\end{equation}
The ratio of outcomes inside the effective lightcone over all possible outcomes $\Delta'$ is
\begin{align}
    \Delta' &= \frac{\mathcal{M}'}{\binom{M+2n-1}{2n}} \\
&\le \sqrt{2}e^{\frac{1}{24n}}\left[ 2^d\left(\frac{2n^{\lambda}D}{\beta d}\right)^{\frac{d}{2}} \frac{2n}{eM}\right]^n \\
&< 2\left[ \frac{2}{ec_1}\left( \frac{8}{\beta d}\right)^{\frac{d}{2}} D^{\frac{d}{2}}n^{1-\gamma + \frac{\lambda d}{2}} \right]^n .
\end{align}
Therefore, for depth $D \le \alpha_1n^{\frac{2(\gamma -1)}{d}-\lambda}$ with a constant $\alpha_1 = c_{1}^{\frac{2}{d}}e^{\frac{2}{d}}\beta d/{2}^{\frac{2}{d}+3}$, $\Delta'$ is exponentially small, which means that probability distribution is concentrated inside the exponentially small portions of outcomes. 

Also, using a similar analysis to the proof of the Theorem \ref{theorem2}, a ratio of truncated outcomes outside the effective lightcones that can have probabilities larger than $\epsilon = \poly(n)^{-1}\frac{(2n)!}{M^{2n}}$ over all possible outcomes, which we previously denoted as $\Gamma$, is upper-bounded by $\Gamma \le \text{poly}(n)\exp(-n^{\lambda})$ with probability $1-\delta$ over the circuit instances. 
This claims that for GBS with depth $D \le \alpha_1n^{\frac{2(\gamma -1)}{d}-\lambda}$, we can well approximate $1 - \Delta'-\Gamma$ of the output probability instances with $1 - \delta$ over the circuit instances, which completes the proof. 
\end{proof}

\section{Probability distributions of random NLHS circuit, for fixed mode number}\label{fixedmodenumber}

To examine the scaling behavior of Haar measure convergence only in terms of photon number, we set mode number $M = 32$ and photon number $N = 8, 16$, and sampled 10000 unitary matrices for each repetition of the random NLHS circuit and global Haar unitary circuit.
Now we fix the input and output mode, as the photon number is comparably large such that the condition to get collision-free outcomes dominantly (i.e., $N = \mathcal{O}(\sqrt{M})$) may no longer hold.
We calculate probabilities for fixed (first $N$ modes over $M$ modes) input and output and plot the probability density function for those values in Fig.~\ref{fig:densitysumfixedmode}.
We find that the number of repetitions required to imitate the distribution of Haar unitary is insensitive to increasing photon number, even for the case that output photon number is comparably larger than $\mathcal{O}(\sqrt{M})$.

\section{Numerical evidence for hiding}\label{hidingevidence}

The hiding theorem is necessary to prove the hardness of the classical simulation of BS; to simplify, the theorem states that hiding random output instances into random circuit instances is possible.
More specifically, for the FBS case, the squared permanent of the $N\times N$ submatrices randomly chosen from $M\times M$ Haar random unitary matrices is very similar to the squared permanent of the random Gaussian matrices $X \sim \mathcal{N}(0,1)_{\mathbb{C}}^{N\times N}$ with normalization factor $\frac{1}{M^N}$, if $M = \omega(N^5)$ is satisfied (conjectured that $M = \omega(N^{2})$ is enough) \cite{aaronson2011computational}.
Also, Ref.~\cite{deshpande2021quantum} suggested hiding property for GBS; for equal squeezing input, the output probability distribution over Haar random circuit instances is similar to the squared hafnian of the product of random Gaussian matrices.

To gather numerical evidence about whether the hiding property holds for the random NLHS circuit,  we compare the results between fixed and random input/output configurations over the random NLHS circuit instances and examine their converging behavior to the distribution of Haar random matrices. 
We also compare those results with the distribution from random Gaussian matrices, that is, $|\text{Per}(X)|^2/M^N$ with $X \sim \mathcal{N}(0,1)_{\mathbb{C}}^{N\times N}$ for FBS \cite{aaronson2011computational}, and $|\text{Haf}(XX^T)|^2/M^N$ with $X \sim \mathcal{N}(0,1)_{\mathbb{C}}^{N\times M}$ for GBS \cite{deshpande2021quantum}; they are believed to be hard to additively estimate on average. 
If the results for fixed and random input/output are close enough to each other, and also close enough to the distributions of Haar random matrices and random Gaussian matrices, it gives evidence that hiding random output probabilities into average-case hard-to-estimate quantity (conjectured) would be possible.

We set mode number $M = 128$ and photon number $N = 10$, and sampled 10000 unitary matrices for each repetition of the random NLHS circuit.
We calculate probabilities both for fixed and random input/output for those matrices, where we fixed input and output modes to the first $N$ modes over $M$ modes.
Also, we calculate probabilities from 10000 Haar unitary matrices and random Gaussian matrices of size $N \times N$ ($N \times M$ for GBS).
We plot the probability density function for those values, which can be checked in Fig.~\ref{fig:hiding}. 
We find that the result demonstrates the manifestation of the hiding property comparably after stacking the random NLHS circuit twice, which is consistent with previous results. 
Both fixed and random output probabilities of random NLHS circuit instances converge to the distributions of Haar random unitary and random Gaussian matrices, suggesting numerical evidence of hiding property at the low-depth regime.

\section{Comparison of frame potential}\label{framepotential}
For given circuit ensemble $\mathcal{E}$, the frame potential $\mathcal{F}_{\mathcal{E}}^{(k)}$ for $k$ moment is given by~\cite{gross2007evenly}
\begin{equation}\label{theoreticalframepotential}
    \mathcal{F}_{\mathcal{E}}^{(k)} = \int_{U,V\in\mathcal{E}}dUdV|\text{Tr}(U^{\dag}V)|^{2k},
\end{equation}
where the theoretical value of frame potential for the Haar random unitary is $\mathcal{F}_{\text{Haar}}^{(k)} = k!$.
By using the Monte-Carlo method, the theoretical value in Eq.~\eqref{theoreticalframepotential} can be approximated numerically by~\cite{liu2022estimating}
\begin{equation}\label{numericalframepotential}
    \mathcal{F}_{\mathcal{E}}^{(k)} \approx \frac{1}{N_{sam}}\sum_{i=1}^{N_{sam}}|\text{Tr}(U_i^{\dag}V_i)|^{2k},
\end{equation}
where the summation is over $N_{sam}$ randomly generated circuits $U_i,V_i$ from the ensemble $\mathcal{E}$.

We compare the frame potential of the random NLHS circuit with the 2-dimensional local parallel random circuit and the 3-dimensional local parallel random circuit for a fixed mode number $M = 64$.
In this case, a single round of the NLHS circuit architecture can be implemented with circuit depth $D = 6$, and a single round of local parallel circuit for $d=2$ ($d=3$) can be implemented with circuit depth $D =4$ ($D = 6$). 
We randomly sample $N_{sam} = 50000$ unitary matrices $U$ and $V$ in Eq.~\eqref{numericalframepotential} for circuit depth starting from $D =4$ (for $d = 2$ local circuit) and $D=6$ (for NLHS circuit and $d=3$ local circuit). 
Using those unitary matrices, we calculate normalized frame potential $\mathcal{F}_{\mathcal{E}}^{(k)}/\mathcal{F}_{\text{Haar}}^{(k)}$ via Monte-Carlo approach as in Eq.~\eqref{numericalframepotential}, for the moment $k = 2, 3, 4$.
In Fig.~\ref{fig:framepotential}, we plot all evaluated values.
We find that all the circuits show the convergence to the moments of Haar random unitary as circuit depth increases.
The random NLHS circuit clearly shows faster convergence behavior with increasing circuit depth compared to other local circuits, which is consistent with the previous results. 
This result suggests numerical evidence of the unitary design property of the random NLHS circuit within a few repetitions.

\end{appendix}

\bibliographystyle{unsrt}
\bibliography{reference}

\end{document}